\documentclass[10pt,twocolumn,twoside]{IEEEtran} 

\usepackage{textcomp}
\usepackage{graphics} 
\usepackage{graphicx}
\usepackage{epsfig} 
\usepackage{times} 
\usepackage{amsmath} 
\usepackage{amssymb}  
\usepackage{bbm}
\usepackage{mathtools}
\usepackage{nccmath}
\usepackage{color,url}
\usepackage{neuralnetwork}
\usepackage{xpatch}
\makeatletter
\xpatchcmd{\linklayers}{\nn@lastnode}{\lastnode}{}{}
\xpatchcmd{\linklayers}{\nn@thisnode}{\thisnode}{}{}

\usepackage[noadjust]{cite}
\usepackage {tikz}
\usetikzlibrary {positioning}
\usepackage {xcolor}
\definecolor {processblue}{cmyk}{0.96,0,0,0}
\usetikzlibrary{shapes,arrows}
\usetikzlibrary{calc}
\usetikzlibrary{decorations.shapes}
\usetikzlibrary{decorations.text,%
  decorations.pathmorphing,%
  decorations.pathreplacing,%
  decorations.footprints,%
  decorations.markings}
\usetikzlibrary{matrix,chains}
 
\usepackage{amsthm}
\usepackage[utf8]{inputenc}
\usepackage[T1]{fontenc}

\usepackage[ruled,linesnumbered]{algorithm2e}
\usepackage{algpseudocode}

\newcommand{\trp}{\intercal}

\input{mysymbol.sty}
\usepackage{needspace}

\theoremstyle{definition}

\newtheorem{assumption}{Assumption}

\newtheorem{remark}{Remark}
\theoremstyle{plain}
\newtheorem{theorem}{Theorem}

\newtheorem{proposition}{Proposition}
\theoremstyle{remark}

\begin{document}
\bstctlcite{BSTcontrol}

\title{Large-Scale Graph Reinforcement Learning in Wireless Control Systems}

\author{Vinicius Lima$^{1}$, Mark Eisen$^{2}$, Konstantinos Gatsis$^{3}$ and Alejandro Ribeiro$^{1}$%
	\thanks{Supported by Intel Science and Technology Center for Wireless Autonomous Systems and ARL DCIST CRA W9111NF-17-2-0181. 
		$^{1}$ V. Lima and A. Ribeiro are with the  Department of Electrical and Systems Engineering, University of Pennsylvania, Philadelphia, PA 19104 USA (e-mail:  vlima@seas.upenn.edu; aribeiro@seas.upenn.edu). 
		$^{2}$ M. Eisen is with Intel Corporation, Hillsboro, OR (e-mail: mark.eisen@intel.com). 
		$^{3}$ K. Gatsis is with the Department of Engineering Science, University of Oxford, Parks Road, Oxford, OX1 3PJ, UK (e-mail: konstantinos.gatsis@eng.ox.ac.uk). 
Preliminary results were presented at the 2020 21st IFAC World Congress \cite{lima_2020_ifac}, which are here expanded upon by generalizing the formulation to include long-term constraints, establishing the permutation invariance and transference of graph reinforcement learning allocation policies in wireless control systems, and presenting extended numerical experiments.
}
}

\maketitle

\begin{abstract}
	Modern control systems routinely employ wireless networks to exchange information between spatially distributed plants, actuators and sensors. With wireless networks  defined by random, rapidly changing transmission conditions that challenge assumptions commonly held in the design of control systems, proper allocation of communication resources is essential to achieve reliable operation.  
	Designing resource allocation policies, however, is challenging, motivating recent works to successfully exploit deep learning and deep reinforcement learning techniques to design resource allocation and scheduling policies for wireless control systems (WCSs). As the number of learnable parameters in a neural network grows with the size of the input signal, deep reinforcement learning may fail to scale, limiting the immediate generalization of such scheduling and resource allocation policies to large-scale systems. The interference and fading patterns among plants and controllers in the network, however, induce a time-varying graph that can be used to construct policy representations based on graph neural networks (GNNs), with the number of learnable parameters now independent of the number of plants in the network. 
	We further establish in the context of WCSs that, due to inherent invariance to graph permutations, the GNN is able to model scalable and transferable resource allocation policies, which are subsequently trained with primal-dual reinforcement learning. 
    Numerical experiments show that the proposed graph reinforcement learning approach yields policies that not only outperform baseline solutions and deep reinforcement learning based policies in large-scale systems, but that can also be transferred across networks of varying size.
\end{abstract}
\begin{IEEEkeywords}
	Wireless Control Systems, Resource Allocation, Graph Neural Networks, Constrained Reinforcement Learning.
\end{IEEEkeywords}


\IEEEpeerreviewmaketitle

\section{Introduction}


Modern control systems routinely employ wireless networks to exchange information between spatially distributed plants, actuators and sensors. The use of wireless networks instead of wired communication makes the installation of components easier and maintenance more flexible, but also adds particular challenges to the design of control and communication policies  \cite{hespanha_survey_2007, park_wireless_2018}. Transmission conditions in wireless channels vary rapidly  and information packets addressed to different users are further subject to destructive interference \cite{tse_2005}. These issues can be alleviated by designing proper allocation policies to distribute limited communication resources across devices in the network. 
Such resource allocation problems usually consist in optimizing some metric of interest while satisfying constraints on resource utilization, cf. e.g. \cite{yaxin_li_scheduling_2001,fattah_leung_scheduling_wl_2002,Eryilmaz_srikant_2007}. 
The nonconvexity and infinite dimensionality of resource allocation problems \cite{Ribeiro2012}, however, makes these problems difficult to solve and motivates the use of heuristic or approximate solutions. A recent focus is the use of deep learning to solve resource allocation in wireless systems \cite{eisen_learning_2018, liang_opc_dnns, liang2019deep, zappone2019wireless}. 

In networked \emph{control} systems, the design of control policies under general communication models has been extensively studied \cite{walsh_stability_2002, nesic_input-output_2004, tabbara_input-output_2005, zhang_recent_advances_NCSs_2016}, as has the problem of resource allocation and user scheduling in WCSs \cite{rehbinder_scheduling_2004,mo_sensor_2011,shi_optimal_2011, gatsis_opportunistic_2015}.
As in the conventional wireless setting, finding optimal resource allocation policies in WCSs usually leads to an intractable optimization problem, motivating the use of heuristics or approximations. 
In place of model-based heuristics, recent efforts have utilized deep reinforcement learning to design resource allocation in WCSs, cf., e.g., \cite{charalambous_resource_2017, demirel_deepcas:_2018,leong_deep_2018,baumann_drl_etc, redder2019deep,lima2020optimal_tsp}.
The scheduling algorithms proposed in \cite{demirel_deepcas:_2018} and \cite{leong_deep_2018}, for example, rely on Deep Q-Network (DQN), a value-based deep RL algorithm. 
An actor-critic algorithm for communication and control  in WCSs is proposed in \cite{baumann_drl_etc}, whereas a constrained reinforcement learning approach to find feasible resource allocation and control policies in WCSs subject to long-term constraints is discussed in \cite{lima2020optimal_tsp}. 

Although learned policies outperform hand-crafted heuristics in a variety of WCSs, they do so in systems that involve small numbers of plants and controllers. This is because they rely on fully connected neural networks, which are difficult to train when the number of input and output variables is large. In the context of wireless communications it has been shown that graph neural networks (GNNs) can leverage the structure of interference patterns to solve resource allocation problems with hundreds of terminals \cite{eisen2019optimal_gnns, liang_graph_2018, shen2020graph, lee2019graph}. The main goal of this paper is to show that analogous success is attainable in WCSs. Specifically, our main contribution is the following:

\begin{list}
      {}
      {\setlength{\labelwidth}{20pt}
       \setlength{\labelsep}{-3pt}
       \setlength{\itemsep}{0pt}
       \setlength{\leftmargin}{20pt}
       \setlength{\rightmargin}{0pt}
       \setlength{\itemindent}{0pt} 
       }

\item [(C) \;] We introduce a constrained graph reinforcement learning approach to design scalable, feasible resource allocation policies in WCSs. 

\end{list}

This approach relies on the use of the time-varying interference \emph{graph} that describes the state of the communication network to construct policies parametrized by GNNs (Section \ref{sec:gnns_perm_inv}). Incorporating the underlying structure of the WCS into the policy parametrization makes the number of learnable parameters independent of the number of agents, which is fundamental to achieve scalability. Policies parametrized by GNNs are incorporated in a \emph{constrained} reinforcement learning problem whose objective is to minimize a discounted quadratic control cost and whose constraints represent the consumption of communication resources (Section \ref{sec:optimal_resource_formulation}). Policies are subsequently trained in a model-free manner using a primal-dual algorithm that alternates between updates of policy parameters via reinforcement learning iterations and updates of dual variables enforcing constraint satisfaction via gradient descent (Section \ref{sec:res_allocation_RL}). Our numerical results show that it is possible to train policies for WCSs with tens of plants and controllers (Section \ref{sec:num_exp}).

To explain the success of GNNs in learning policies for WCS we explore their transference properties. Our first investigation concerns the effects of node labeling. We therefore consider systems whose respective plant dynamics and wireless channel statistics are permutations of each other (Section \ref{subsec:ra_wcss_invariance}). Comparing the corresponding learning problems leads to the first property of GNNs that we establish

\begin{list}
      {}
      {\setlength{\labelwidth}{26pt}
       \setlength{\labelsep}{-3pt}
       \setlength{\itemsep}{0pt}
       \setlength{\leftmargin}{26pt}
       \setlength{\rightmargin}{0pt}
       \setlength{\itemindent}{0pt} 
       }

\item [(P1) \;] We prove that a GNN policy is optimal for a constrained graph reinforcement learning problem if and only if it is optimal for any other problem in which respective plant dynamics and wireless channel statistics are permutations of each other (Theorem \ref{theo:filter_invariance} in Section \ref{subsec:ra_wcss_invariance}).

\end{list}

This property means that parameters that are optimal for a given WCS are also optimal for all of its permutations. That is, that optimal GNN policies are permutation invariant. This fact facilitates transference of learned resource allocation policies across similar networks. This is an important property because we use different networks during training and execution. This implies that trained GNN are evaluated in networks that are different from those in which they were trained. System realizations for training and testing are drawn with the same statistics but their specific realizations are different. Property (P1) implies that if the networks that we draw are not far from permutations of each other we must observe successful transference of trained GNNs. Our numerical experiments corroborate that this is true (Sections \ref{sec_numericals_cellular} and \ref{sec_numericals_adhoc}).

Our second investigation concerns transference \emph{at scale}. This idea concerns the execution of a trained GNN policy in a WCS that is much larger than the WCS in which it was trained:

\begin{list}
      {}
      {\setlength{\labelwidth}{26pt}
       \setlength{\labelsep}{-3pt}
       \setlength{\itemsep}{0pt}
       \setlength{\leftmargin}{26pt}
       \setlength{\rightmargin}{0pt}
       \setlength{\itemindent}{0pt} 
       }

\item [(P2) \;] We numerically demonstrate that learned policies can be successfully transferred across networks of varying size (Section \ref{subsec:transference_sims}). 

\end{list}

That is, one can learn an allocation policy for a given WCS, and successfully deploy that policy on larger systems. We specifically demonstrate training in networks with tens of nodes that we successfully transfer to networks with hundreds of nodes.

\section{Resource Allocation in WCSs}
\label{sec:optimal_resource_formulation}


Consider a system made up by $m$ independent control loops sharing a common wireless medium to communicate with remote base stations --- which in turn communicate with each other over a wired connection, as shown in Figure \ref{fig:communication_structure}. Each control plant $i$ can be described by a discrete, time-invariant model $f^{(i)}(\cdot, \cdot) \, : \, \mathbb{R}^{p} \times \mathbb{R}^q \to \mathbb{R}^p $ mapping  a current state vector $x_t^{(i)} \in \mathbb{R}^p$ and corresponding control input $u_t^{(i)} \in \mathbb{R}^q$ to the next state of the system. Furthermore, each of those plants is affected by some zero-mean random noise $w^{(i)}_t \in \mathbb{R}^p$ with covariance matrix $W^{(i)} \in \mathbb{R}^p$  standing for eventual disturbances and unmodeled dynamics, leading to
\begin{equation}
x^{(i)}_{t + 1} = f^{(i)}(x_t^{(i)}, u_t^{(i)}) + w^{(i)}_t. 
\label{eq:ind_plant_dyn}
\end{equation}
Actuation signals $g(x_t^{(i)})$ are computed at the corresponding remote controller co-located with the base stations, based on observations of states sent by the plants,
\begin{equation}
\tilde{x}^{(i)}_{t} = x^{(i)}_{t} + w_{o, t}^{(i)},
\label{eq:observation_noise}
\end{equation}
with $w_{o, t}^{(i)} \in \mathbb{R}^p$ representing a zero-mean observation noise with covariance matrix $W_{o}^{(i)} \in \mathbb{R}^p$.
As wireless networks are prone to packet drops, the control signal $u^{(i)}_t$ \textit{received} by plant $i$ will depend on the channel transmission conditions. If transmission conditions are favorable and the transmission of the packet sent by the remote controller is successful, the plant executes its intended control action. When transmission fails, however, we assume the plant does not execute any action. 

In other words, significant noise in the wireless channel will cause the system to operate in open loop, in which it cannot be directly controlled. The probability of closing the feedback loop will depend upon the current \emph{fading state} in the wireless channel, the \emph{transmission power} allocated to the control signal sent by the base station,  as well as \emph{interference} caused by transmissions made by  other control loops sharing the wireless medium. We may directly or indirectly control power and interference levels in response to the changes in the fading state, which is itself out of our control.
%
%
\tikzstyle{block} = [draw, fill=blue!20, rectangle, 
    minimum height=1cm, minimum width=5em]
\tikzstyle{largerblock} = [draw, fill=blue!20, rectangle, 
    minimum height=1cm, minimum width=18em]
\tikzstyle{sum} = [draw, fill=blue!20, circle, node distance=1cm]
\tikzstyle{input} = [coordinate]
\tikzstyle{output} = [coordinate]
\tikzstyle{pinstyle} = [pin edge={to-,thin,black}]
\tikzset{edge/.style = {->,> = latex'}}
\begin{figure*}
  \centering
  \begin{tikzpicture}[scale=.9,
	path/.style={
		->,
		>=stealth,
	},
	decoration={
		markings,
		mark=at position 0.59cm*1 with {\arrow[black]{stealth}},
		mark=at position 0.59cm*2 with {\arrow[black]{stealth}},
		mark=at position 0.59cm*3 with {\arrow[black]{stealth}},
		mark=at position 0.59cm*4 with {\draw (0,0) circle;},
		mark=at position 0.6cm*5-0.01cm with {\arrow[black]{stealth}},
		mark=at position 6mm*6-0.1 with {\arrow[black]{stealth}},
		mark=at position 7*6mm-0.1 with {\arrow[black]{stealth}},
		mark=at position 8*6mm-0.1 with {\draw (0,0) circle;},
		mark=at position 9*6mm-0.1 with {\draw (0,0) circle;},
		mark=at position 10*6mm-0.1 with {\draw (0,0) circle;},
		mark=at position 11*6mm-0.1 with {\arrow[black]{stealth}},
		mark=at position 12*6mm-0.1 with {\arrow[black]{stealth}},
		mark=at position 13*6mm-0.1 with {\arrow[black]{stealth}},
		mark=at position 14*6mm-0.1 with {\draw (0,0) circle;},
		mark=at position 15*6mm-0.1 with {\arrow[black]{stealth}},
		mark=at position 16*6mm-0.1 with {\arrow[black]{stealth}},
		mark=at position 17*6mm-0.1 with {\arrow[black]{stealth}},
		mark=at position 18*6mm-0.1 with {\arrow[black]{stealth}},
		mark=at position 19*6mm-0.1 with {\draw (0,0) circle;},
		mark=at position 20*6mm-0.1 with {\arrow[black]{stealth}},
		mark=at position 21*6mm-0.1 with {\arrow[black]{stealth}},
		mark=at position 22*6mm-0.1 with {\arrow[black]{stealth}},
		mark=at position 23*6mm-0.1 with {\draw (0,0) circle;},
	}, radius=1pt,]

\foreach \x/\xtext in {300/2}
	\draw (\x:2.4) node {};
	
	\node[circle,draw=black, fill=red!20, inner sep=0pt,minimum size=15pt ] (receiver) at (-6,0) {$BS_1$};
	
	\node[circle,draw=black, fill=red!20, inner sep=0pt,minimum size=15pt] (receiver2) at (5,0) {$BS_n$};
	
	\node[circle,draw=black, fill=red!20, inner sep=0pt,minimum size=15pt] (receiver3) at (-1,0) {$BS_2$};
	
	\node[circle,draw=black, fill=blue!20, inner sep=0pt,minimum size=10pt] (plant1) at (-7.5,1.5) {$p^1_1$};
	
	\node[circle,draw=black, fill=blue!20, inner sep=0pt,minimum size=10pt] (plant2) at (-7.5,-1.5) {$p^1_2$};
	
\node[circle,draw=black, fill=blue!20, inner sep=0pt,minimum size=10pt] (plant3) at (-4.5,1.5) {$p^1_k$};

\node[circle,draw=black, fill=blue!20, inner sep=0pt,minimum size=10pt] (plant4) at (-4.5,-1.5) {$p^1_3$};
	
	\node[circle,draw=black, fill=blue!20, inner sep=0pt,minimum size=10pt] (plantm1) at (3.5,1.5) {$p^n_1$};
	
	\node[circle,draw=black, fill=blue!20, inner sep=0pt,minimum size=10pt] (plantm) at (6.5,-1.5) {$p^n_3$};
	
		\node[circle,draw=black, fill=blue!20, inner sep=0pt,minimum size=10pt] (plantm2) at (3.5,-1.5) {$p^n_2$};
	
	\node[circle,draw=black, fill=blue!20, inner sep=0pt,minimum size=10pt] (plantm3) at (6.5,1.5) {$p^n_k$};
	
		\node[circle,draw=black, fill=blue!20, inner sep=0pt,minimum size=10pt] (plant21) at (-2.5,1.5) {$p^2_1$};
	
	\node[circle,draw=black, fill=blue!20, inner sep=0pt,minimum size=10pt] (plant22) at (.5,-1.5) {$p^2_3$};
	
	\node[circle,draw=black, fill=blue!20, inner sep=0pt,minimum size=10pt] (plant23) at (-2.5,-1.5) {$p^2_2$};
	
	\node[circle,draw=black, fill=blue!20, inner sep=0pt,minimum size=10pt] (plant24) at (.5,1.5) {$p^2_k$};
	
	\path[thick](receiver) edge (receiver3);
	
	\draw[edge](plant1) to[bend left=15] 
	node[above right, pos=0.3]{$x^{(1)}$} (receiver);
	\draw[edge, dashed](receiver) to[bend left=15] 
	node[left, pos=0.1]{$u^{(1)}$} (plant1);
	
	\draw[edge](plant2) to[bend left=15] 
	node[left, pos=0.2]{$x^{(2)}$} (receiver);
	\draw[edge, dashed](receiver) to[bend left=15]  
	node[right, pos=0.6]{$u^{(2)}$} (plant2);
	
	\draw[edge](plant3) to[bend left=15] 
	node[right, pos=0.4]{$x^{(k)}$} 
	(receiver);
	\draw[edge, dashed](receiver) to[bend left=15]  
	 node[left, pos=0.6]{$u^{(k)}$} (plant3);
	 
	 \draw[edge](plant4) to[bend left=15] 
	 node[left, pos=0.2]{$x^{(3)}$} 
	(receiver);
	\draw[edge, dashed](receiver) to[bend left=15]  
	 node[above right, pos=0.6]{$u^{(3)}$} (plant4);
	
	\node[circle, draw=none] (auxrec2) at (1.5,0) {};
	\node[circle, draw=none] (auxrec3) at (2.5,0) {};
	\node[circle, draw=none] (auxrec1) at (-7.5,0) {};
	\node[circle, draw=none] (auxrecm) at (6.5,0) {};
\path (auxrec2) -- node[auto=false]{\ldots} (auxrec3);	
\path[thick](receiver3) edge (auxrec2);
\path[thick](auxrec3) edge (receiver2);
\path[thick](receiver2) edge (auxrecm);
\path[thick](auxrec1) edge (receiver);
	
	\draw[edge](plantm1) to[bend left=15] 
	node[above right, pos=0.2]{$x^{(nk -3)}$} (receiver2);
	\draw[edge, dashed](receiver2) to[bend left=15]  
	 node[below left, pos=0.6]{$u^{(nk-3)}$} (plantm1);
	
	\draw[edge](plantm3) to[bend left=15] 
	node[right, pos=0.2]{$x^{(nk)}$} (receiver2);
	\draw[edge, dashed](receiver2) to[bend left=15]  
	node[left, pos=0.6]{$u^{(nk)}$} (plantm3);
	
	\draw[edge](plantm) to[bend left=15] 
	node[ left, pos=0.2]{$x^{(nk-1)}$} (receiver2);
	\draw[edge, dashed](receiver2) to[bend left=15]  
	node[right, pos=0.6]{$u^{(nk-1)}$} (plantm);
	
	\draw[edge](plantm2) to[bend left=15] 
	node[left, pos=0.5]{$x^{(nk-2)}$} (receiver2);
	\draw[edge, dashed](receiver2) to[bend left=15]  
	node[below right, pos=0.9]{$u^{(nk-2)}$} (plantm2);
	
		\draw[edge](plant21) to[bend left=15] 
	node[above right, pos=0.2]{$x^{(2k -3)}$} (receiver3);
	\draw[edge, dashed](receiver3) to[bend left=15]  
	node[below left, pos=0.6]{$u^{(2k-3)}$} (plant21);
	
	\draw[edge](plant24) to[bend left=15] 
	node[right, pos=0.2]{$x^{(2k)}$} (receiver3);
	\draw[edge, dashed](receiver3) to[bend left=15]  
	node[left, pos=0.6]{$u^{(2k)}$} (plant24);
	
	\draw[edge](plant22) to[bend left=15] 
	node[ left, pos=0.2]{$x^{(2k-1)}$} (receiver3);
	\draw[edge, dashed](receiver3) to[bend left=15]  
	node[right, pos=0.6]{$u^{(2k-1)}$} (plant22);
	
	\draw[edge](plant23) to[bend left=15] 
	node[left, pos=0.5]{$x^{(2k-2)}$} (receiver3);
	\draw[edge, dashed](receiver3) to[bend left=15]  
	node[below right, pos=0.9]{$u^{(2k-2)}$} (plant23);

\end{tikzpicture}
\caption{Wireless control system architecture. A collection of $m$ independent plants shares a wireless communication network to communicate with remote controllers co-located with the base stations, which regulate the communication between plants and controllers.}
\label{fig:communication_structure}
\end{figure*}
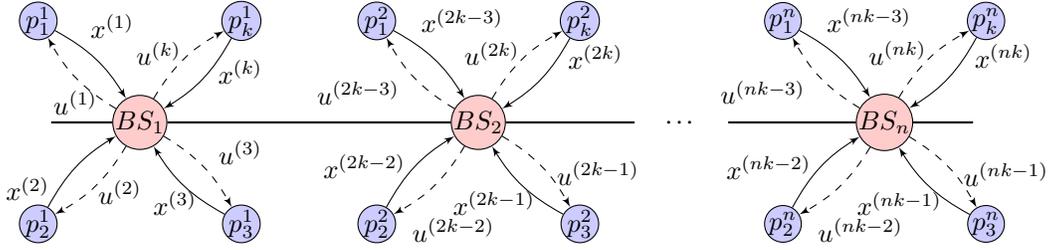

 More precisely, fading stands for the rapidly changing transmission conditions in a wireless channel \cite[ch. 2]{tse_2005}. Fading conditions in channel $i$ can be represented by a random variable  $h^{(ii)} \in \mathcal{H} \subseteq \mathbb{R}_+$ drawn from some probability distribution $\chi(h)$. Due to multiple plants communicating over the same wireless network, significant interference may be caused by concurrent transmissions. The fading state of the interference link between the base station associated with plant $i$ and the receiver associated with plant $j$ is given by random variables  $h^{(ij)}$. The overall fading conditions in the network can be aggregated in a matrix $H$ describing the state of the wireless network at a given point in time,
\begin{equation}
H = \begin{bmatrix}
h^{(11)} & \dots & h^{(1m)} \\
\vdots & \ddots & \vdots \\
h^{(m1)} & \dots & h^{(mm)} 
\end{bmatrix}.
\label{eq:interference_matrix}
\end{equation}

The reliability of the wireless channel does not depend only on the wireless network state, $H$, but also on the power levels used in each of the transmissions. 
  Based on the observations received from the plants \eqref{eq:observation_noise}, as well as current fading and interference conditions in the wireless network, summarized in the interference matrix, base stations compute a power allocation decision $\alpha_{t}$ sampled from an allocation policy parametrized by some parameter vector $\theta$,
  \begin{equation}
  \alpha_{t} = \alpha(\tilde{x}_t, H_t; \theta).
  \label{eq:uplink_param}
  \end{equation}

 We denote by $\alpha^{(i)} \coloneqq [ \alpha_{t} ]_i \in \reals_+$ a power level used in the transmission of the control packet sent to plant $i$.
In particular, for the control signal sent to plant $i$, we can define the signal to interference plus noise ratio (SINR)  $\xi^{(i)}: \reals_{+}^m \times \mathcal{H}^{m^2}  \rightarrow \mathbb{R}_+$, 
\begin{equation}
\xi^{(i)} = \frac{h^{(ii)}\alpha^{(i)}}{ \sigma^2 + \sum_{j \neq i} h^{(ij)}\alpha^{(j)} },
\label{eq:signal-to-noise-interference}
\end{equation}
where $\sigma^2$ is the variance of a standard AWGN channel. The SINR $\xi^{(i)}$ gives the \emph{effective} channel quality of the link between plant $i$ and its remote controller after the effects of transmission power and interference are taking into account. Observe that higher power levels increase the SINR in the direct link, but cause greater interference for other transmissions in the network. Mitigating the interference across many devices through power control remains an active problem in wireless communications \cite{shi_wmmse_2011, naderializadeh_itlinq_2014}.

From the control system standpoint, the effective channel quality as measured by the SINR directly specifies current performance of a wireless control loop. In particular, we define a general function
 $v: \reals_{+}^m \rightarrow [0,1]^m$ that, given the SINR $\xi^{(i)}$ in a wireless channel, returns the probability of  plants closing the feedback loop. 
  This is typically modeled as a sigmoidal-shaped function, but we do not assume such a model is known. As the SINR ultimately defines the probability of a plant successfully executing the prescribed control action, the  control signal \emph{received} by plant $i$ in \eqref{eq:ind_plant_dyn} is  defined as
\begin{equation}
 u^{(i)}_{t} = 
 \begin{cases}
    g(\tilde{x}_t^{(i)}), \text{ w.p. } v( \xi^{(i)}),  \\
    0, \text{ w.p. } 1 - v(\xi^{(i)}). 
 \end{cases}
 \label{eq:plant_switched_prob}
\end{equation}

Contrary to traditional wireless networking problems, the fundamental objective here is to keep plants operating around an equilibrium point --- assumed to be zero without loss of generality --- while respecting constraints on communication resources. 
The objective of the resource allocation problem can then be represented by a quadratic cost that penalizes large deviations of the plant states from the equilibrium point. 
Here, we search for a resource allocation policy $ \alpha(\tilde{x}_t, H_t; \theta) \, \, : \mathbb{R}^{m p} \times \mathbb{R}^{m^2} \rightarrow \mathcal{A} \subseteq \mathbb{R}_+^m$ that, given current channel conditions aggregated in $H_t := [h_t^{(11)}, h_t^{(1m)}, \hdots, h_t^{(mm)}]$ and estimates of plant states $\tilde{x}_t := [\tilde{x}_t^{(1)}, \hdots, \tilde{x}_t^{(m)}]$, returns an allocation vector $\alpha \in \mathcal{A} \subseteq \mathbb{R}^m_+$, where the set $\mathcal{A}$ defines the space of possible allocation decisions as given by the communication model.
The resource allocation problem consists then in finding the allocation policy that minimizes the cost of operating the plants over a finite horizon $T$ starting from some initial state $x_0$, while respecting constraints $l(\alpha): \mathbb{R}^m \to \mathbb{R}^r$ on the power used to send control signals back to the plants, 
\begin{equation}
\begin{aligned}
\theta := &\arg \min_{\theta} J(\theta) =  \mathbb{E}_{x_0}^{ \alpha(\cdot; \theta)} \left[ \sum_{t = 0}^T \gamma^t c(x_t)\right], \\
\text{s.t. } & L(\theta) = \mathbb{E} \left[ \sum_{t=0}^T \gamma^t l(\alpha)  \right] \leq 0, \\
& \alpha \in \mathcal{A}, 
\label{eq:constrained_optimal_prob}
\end{aligned}
\end{equation}
with $c(x_t) =  \sum_{i = 1}^m x_t^{(i)^T}Qx_t^{(i)} $ the one-step cost, $\gamma \in [0, 1]$ a discount factor and  $Q \geq 0$ a weight matrix. We note that, although performance of a controller is often measured by a quadratic cost that penalizes both large deviations from the equilibrium point and large control efforts, as in the classical linear quadratic control formulation, the control policy here is known a priori. Thus, there is no need to additionally penalize large control efforts in the quadratic cost given in \eqref{eq:constrained_optimal_prob}.

 The optimization problem in  \eqref{eq:constrained_optimal_prob} involves finding the resource allocation function $\alpha(x,H)$ that minimizes the operation cost of the plants while satisfying the resource constraints. As the optimization problem has infinite dimensionality, it is generally intractable to find optimal solutions even if we restrict our attention to systems with a low number of plants and with short optimization horizons. Moreover, finding an optimal policy directly in \eqref{eq:constrained_optimal_prob} necessarily requires explicit knowledge of the plant dynamics and communication models in \eqref{eq:plant_switched_prob}, which are often unavailable in practice.
The challenging nature of the problem 
 motivated recent works \cite{demirel_deepcas:_2018,leong_deep_2018,baumann_drl_etc, redder2019deep, lima2020optimal_tsp} to use deep reinforcement learning to design resource allocation policies for wireless control systems.

\section{Graph Neural Networks}
\label{sec:gnns_perm_inv}

Reinforcement learning (RL) is a mathematical framework to handle sequential decision problems. At each time step, an agent executes some action $a_t \in \mathcal{A}$ sampled from a stochastic policy $\pi(a|s) \in \Pi$, observes the resulting state of the system, $s_t \in \mathcal{S}$, receives a one-step cost $r_t$ from the environment and then tries to find a policy that minimizes the cumulative cost of those transitions \cite{sutton_reinforcement_learning}. 
 As looking for policies directly is usually infeasible, one approximates the stochastic policy with some parametrization, with that parametrization corresponding to a neural network --- high capability approximators \cite{hornik_1989} --- in the field of \emph{deep} reinforcement learning. The sucess of deep reinforcement learning, in turn, motivated the use of deep RL algorithms to design resource allocation policies for wireless control systens, cf., e.g., \cite{demirel_deepcas:_2018, leong_deep_2018,baumann_drl_etc, redder2019deep, lima2020optimal_tsp}. 

A standard neural network consists of a series of computational layers where each unit in layer $l$ computes a linear combination  of the outputs of layer $l - 1$, and then applies a pointwise nonlinear transformation on top of that linear combination. 
Each hidden layer $l$ is composed of hidden units than can be computed by 
 \begin{equation}
z_{l} = \phi( C_{l} z_{l-1} + b_l),
\label{eq:nn_hidden_units}
 \end{equation}
 with $\phi(\cdot)$ a nonlinearity and the matrices $C_l$, $b_l$ aggregating the weights of the linear combination in that layer.
 Combining the successive computational layers, the output of the neural network is then given by 
\begin{equation}
y_{\text{NN}}(z_0) = \phi \left(C_L  \phi \left( \dots \phi \left( C_1 z_0 + b_1 \right) \right) + b_L  \right). 
\label{eq:nn_output}
\end{equation}
For the allocation policy in \eqref{eq:constrained_optimal_prob}, the input $z_0$ corresponds to, e.g., estimates of the plants states and interference conditions in the network, whereas the output of the neural network characterizes the corresponding allocation policy.
 
 Multilayer neural networks are known near-universal approximators \cite{hornik_1989}, but rely on a rather large number of learnable parameters. As evidenced by equations  \eqref{eq:nn_hidden_units} - \eqref{eq:nn_output}, the learnable parameters correspond to the weights and biases used in linear combinations at each hidden unit in the network, that is, $\theta = [C_1; b_1; \hdots; C_L; b_l] $. As the dimension of the input to the neural network grows, so does the number of learnable parameters. In particular, let $K_l$ be the number of hidden units in layer $l$. For the resource allocation problem, the overall number of learnable parameters is given by
\begin{equation}\label{eq:dnn_dim}
r_{\text{NN}} = m(m + n_f)K_1 + \sum_{l = 2}^L K_l K_{l - 1},
\end{equation} 
scaling linearly with the dimension of the input associated with each plant, $n_f$, and quadratically with the number of plants in the network, $m$.
Although the resource allocation problem \eqref{eq:constrained_optimal_prob} lacks the type of spatial or temporal regularity that allows convolutional neural networks to circumvent this dimensionality issue in applications such as image processing, the communication model  \eqref{eq:interference_matrix} does define an underlying graph structure for the optimization problem. This suggests the use of \emph{graph} neural networks (GNNs) to parameterize the resource allocation policies, as we discuss next. 

\subsection{Graph Neural Networks}
\label{sec:gnn_basics}
GNNs can be viewed as a generalization of the popular convolutional neural network (CNN) model. In CNNs, the linear operations used in standard deep neural networks are replaced with linear convolutional filters. 
This more controlled structure significantly reduces the number of overall parameters the model must learn during training, since the algorithm now learns the coefficients of the corresponding convolutional filters and not the weights of the linear combinations computed at every neuron in the network. 
While the convolutions employed by CNNs are naturally suited for processing of temporal or spatial data, the same does not hold true for inputs without such a regular structure. The WCS architecture considered here, however, nonetheless contains structure embedded in the fading and interference patterns $H$ that can be incorporated into the policy parameterization --- namely, this structure may be represented by a graph. 

GNNs generalize the CNN model by replacing the standard convolutional filter with a \emph{graph} convolutional filter \cite{Gama_2019}. 
 For a graph $G = (\mathcal{V}, \mathcal{E}, \mathcal{W})$ with node and edge sets $\mathcal{V} = \{1, \dots, N \}$, $\mathcal{E} = \{ (i, j); i,j \in \mathcal{V} \} $ and weight function $\mathcal{W} \, : \, \mathcal{E} \to \mathbb{R} $, let then the graph shift operator (GSO) be defined as a matrix $S \in \mathbb{R}^{N \times N}$ that reflects the sparsity of the graph, that is, $S_{ij} = 0$ if $i \neq j$ and $(i,j) \notin \mathcal{E}$. 
Let also $y = [y^{(1)}, \dots, y^{(N)}]$ a graph signal with components $y^{(i)} \in \mathbb{R}^{n_f}  $, 
 $i \in \mathcal{V}$. A graph \emph{convolution} can then be defined as a weighted sum of shifted versions of the graph signal,
\begin{equation}
	z = \sum_{k=0}^{K - 1} S^k y \Psi_k,
	\label{eq:graph_conv}
\end{equation}
producing another graph signal $z\in \mathbb{R}^{m \times g}$ with $g$ features. The matrix $\Psi_k \in \mathbb{R}^{n_f \times g}$ aggregates the \emph{filter taps} $[\Psi_k]_{fg} = \psi_k^{fg}$ used to modulate information received by the $k$-hop neighborhood of each node.
%
\tikzstyle{block} = [draw, fill=blue!20, rectangle, 
    minimum height=1cm, minimum width=5em]
\tikzstyle{largerblock} = [draw, fill=blue!20, rectangle, 
    minimum height=1cm, minimum width=18em]
\tikzstyle{sum} = [draw, fill=blue!20, circle, node distance=1cm]
\tikzstyle{input} = [coordinate]
\tikzstyle{output} = [coordinate]
\tikzstyle{pinstyle} = [pin edge={to-,thin,black}]
\tikzset{edge/.style = {->,> = latex'}}
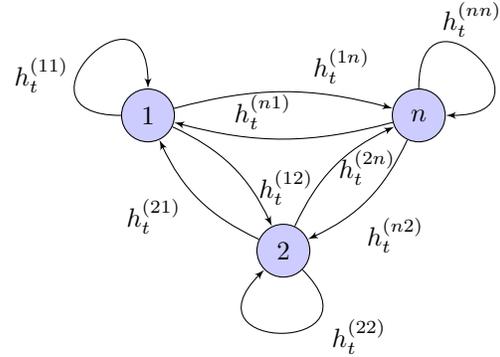
\begin{figure}
\centering
  \def \thisplotscale {4}
\def \unit {\thisplotscale cm}
%

\begin{tikzpicture}[scale=0.9,
	path/.style={
		->,
		>=stealth,
	},
	decoration={
		markings,
		mark=at position 0.59cm*1 with {\arrow[black]{stealth}},
		mark=at position 0.59cm*2 with {\arrow[black]{stealth}},
		mark=at position 0.59cm*3 with {\arrow[black]{stealth}},
		mark=at position 0.59cm*4 with {\draw (0,0) circle;},
		mark=at position 0.6cm*5-0.01cm with {\arrow[black]{stealth}},
		mark=at position 6mm*6-0.1 with {\arrow[black]{stealth}},
		mark=at position 7*6mm-0.1 with {\arrow[black]{stealth}},
		mark=at position 8*6mm-0.1 with {\draw (0,0) circle;},
		mark=at position 9*6mm-0.1 with {\draw (0,0) circle;},
		mark=at position 10*6mm-0.1 with {\draw (0,0) circle;},
		mark=at position 11*6mm-0.1 with {\arrow[black]{stealth}},
		mark=at position 12*6mm-0.1 with {\arrow[black]{stealth}},
		mark=at position 13*6mm-0.1 with {\arrow[black]{stealth}},
		mark=at position 14*6mm-0.1 with {\draw (0,0) circle;},
		mark=at position 15*6mm-0.1 with {\arrow[black]{stealth}},
		mark=at position 16*6mm-0.1 with {\arrow[black]{stealth}},
		mark=at position 17*6mm-0.1 with {\arrow[black]{stealth}},
		mark=at position 18*6mm-0.1 with {\arrow[black]{stealth}},
		mark=at position 19*6mm-0.1 with {\draw (0,0) circle;},
		mark=at position 20*6mm-0.1 with {\arrow[black]{stealth}},
		mark=at position 21*6mm-0.1 with {\arrow[black]{stealth}},
		mark=at position 22*6mm-0.1 with {\arrow[black]{stealth}},
		mark=at position 23*6mm-0.1 with {\draw (0,0) circle;},
	}, radius=1pt,
	baseline=(receiver.base)]

\foreach \x/\xtext in {300/2}
	\draw (\x:2.4) node {};
	
	\node[circle,draw=black, fill=blue!20, inner sep=0pt,minimum size=20pt] (plant1) at (-4.5,1.) {$1$};
	
	\node[circle,draw=black, fill=blue!20, inner sep=0pt,minimum size=20pt] (plant2) at (-2.5,-1) {$2$};
	
\node[circle,draw=black, fill=blue!20, inner sep=0pt,minimum size=20pt] (plant3) at (-.5,1) {$n$};

	\draw[edge](plant1) to[bend left=20] 
	node[right, pos=0.7]{$h_t^{(12)}$} (plant2);
	\draw[edge](plant2) to[bend left=20] 
	node[below left, pos=0.6]{$h_t^{(21)}$} (plant1);
	
\draw[edge](plant1) to[bend left=15] 
	node[above right, pos=0.6]{$h_t^{(1n)}$} (plant3);
	\draw[edge](plant3) to[bend left=15] 
	node[above, pos=0.6]{$h_t^{(n1)}$} (plant1);
	
	\draw[edge](plant3) to[bend left=20] 
	node[below right, pos=0.6]{$h_t^{(n2)}$} (plant2);
	\draw[edge](plant2) to[bend left=20] 
	node[below, pos=0.8]{$h_t^{(2n)}$} (plant3);
	
	\draw[edge](plant1) to[out=180,in=90,loop] 
	node[above left, pos=0.2]{$h_t^{(11)}$} (plant1);
	
	\draw[edge](plant2) to[out=-45,in=-135,loop] 
	node[below right, pos=0.2]{$h_t^{(22)}$} (plant2);
	
	\draw[edge](plant3) to[out=90,in=0,loop] 
	node[above right, pos=0.2]{$h_t^{(nn)}$} (plant3);

\end{tikzpicture}
\caption{The communication graph of the wireless control system defines a time-varying interference graph  over which we  can construct a GNN.}
\label{fig:interference_graph_pic}
\end{figure}
%

Combining graph convolutions and pointwise nonlinear operations yields graph \emph{neural networks}. Each layer $l$ in a GNN takes as input a graph signal $y_l$ produced by the previous layer and outputs a graph signal $y_{l + 1}$ computed by a graph convolution followed by a nonlinear operation,
\begin{equation}
y_{l + 1} = \phi_l \left(\sum_{k=0}^{K - 1} S^k y_l \Psi_{lk}  \right).
\label{eq:gnn_conv_layer}
\end{equation}
The nonlinear operation $\phi_l$ may be any function that respects the local structure of the GSO $S$ \cite{Gama_2019}. 

The interference matrix representing the state of the wireless network underlying the control system at some time instant $t$, $H_t$, can be used to define a graph $\mathcal{G}_t$ with nodes $\ccalV = \{1,\hdots,m\}$ given by the $m$ plants and edges corresponding to the interference between the transmitter associated to plant $i$ and the receiver associated to plant $j$, that is, $\mathcal{W}_t((i,j)) := h^{(ij)}_t$ --- see Figure \ref{fig:interference_graph_pic}. The GSO here corresponds to the interference matrix itself, $S_t = H_t$. 
 Observe that while in standard applications of GNNs, e.g. \cite{Gama_2019}, the graph $\mathcal{G}$ is fixed, here the graph defined by the interference model in \eqref{eq:interference_matrix} is randomly distributed. Hence we make use of the notion of \emph{random edge} GNNs (REGNNs) introduced by \cite{eisen2019optimal_gnns}. If we denote by $z_0$ the input graph signal --- made up, for example, by the current state of the plants in the wireless control system --- then we can define the REGNN as
\begin{multline}
y_{\text{GNN}}(z_0) = \\ \phi_L \left(\sum_{k=0}^{K_L - 1} H_t^k   \left( \dots  \phi_1 \left( \sum_{k=0}^{K_1 - 1} H_t^k z_0 \Psi_{1k} \right) \dots  \right)\Psi_{Lk} \right).
\label{eq:gnn_multilayer}
\end{multline}
This output may then be used to parameterize the policy distribution in \eqref{eq:uplink_param} --- e.g. success probability of a Bernoulli distribution --- where $\theta := [\Psi_1, \hdots, \Psi_L]$ contains the filter coefficients that define the REGNN. 
One needs to learn only the coefficients of the graph filters used at each hidden layer; letting $K_l$ the filter length and $F_l$ the number of features at each layer $l$, 
the overall number of learnable parameters  is
\begin{equation}\label{eq:gnn_dim}
r_{\text{GNN}} = \sum_{l = 1}^L K_l F_l F_{l-1},
\end{equation}
with $F_0 = n_f$, the number of features of the input signal, and $F_L$ the number of features of the output signal.
Observe in \eqref{eq:gnn_dim} that, contrary to standard neural networks \eqref{eq:dnn_dim}, the parameter dimension $r_{\text{GNN}}$ is independent of the number of plants, $m$, and thus structurally suitable to large control networks both in terms of its scalable dimensionality and possible transferability to alternative networks of varying size. This latter case is of particular interest in the practical design of learning solutions, which we explore in greater detail in the proceeding section.


\subsection{Permutation Equivariance of (RE)GNNs}

Training GNNs boils down to learning the filter coefficients $\psi_l$ used at each layer of the GNN. As the filter coefficients do not depend on a particular graph, they can be trained in a given graph but applied to any other graph shift operator $S$ to construct a graph convolution. Hence, a trained GNN can be deployed on other networks, making them \emph{transferable}. Naturally, transferability of filter taps does not necessarily imply that the trained GNN will achieve the same performance on the new network. While we investigate more general instances of transferability with numerical experiments in Section \ref{sec:num_exp}, we first turn our attention into \emph{permutations} or node reorderings. 
As node reorderings do not alter the structure of the graph, it is not unreasonable to expect that filters learned in a graph can be applied to its permutations without loss of performance. 

Let then the set of permutation matrices be defined as
\begin{equation}
\mathcal{P} = \left\{ P \in \{ 0, 1\}^{m \times m}: P \mathbf{1} = 1, P^\trp \mathbf{1} = 1 \right\}.
\end{equation}
For any matrix $P$ in $\mathcal{P}$, the product $P^\trp v$ reorders the entries of vector $v$, and, accordingly, the product $P^\trp M P$ reorders the rows and columns of matrix $M$. 
One can show that GNNs are permutation equivariant, in the sense that node reorderings lead to a similar permutation of the output of a GNN \cite{Gama_2019}. 
Initially discussed for GNNs with fixed topology, the notion of permutation equivariance was later extended to random edge GNNs, as shown in \cite{eisen2019optimal_gnns}.  
\begin{proposition}[Eisen et al  {\cite[Prop.~2]{eisen2019optimal_gnns}} ]
\label{prop:eisen-regnns-inv}
Consider graphs $H$ and $\hat{H}$ along with signals $z$ and $\hat{z}$ such that for some permutation matrix $P$ we have $\hat{H} = P^\trp H P$ and $\hat{z} = P^\trp z$. The outputs of a REGNN with filter tensor $\theta$ to the pairs $(H, z)$ and $(\hat{H}, \hat{z})$ are such that 
\begin{equation}
y(\hat{H}, \hat{z}; \theta) = P^\trp y(H, z; \theta).  
\end{equation}
\end{proposition}
Proposition \ref{prop:eisen-regnns-inv} states that, if we reorder the nodes of the graph with some permutation matrix $P$, the output of the GNN will be reordered accordingly. Hence, we can expect permutations not to affect the performance of a GNN. 
Although it is not immediately clear if the permutation equivariance structure of (RE)GNNs is shared by the resource allocation problem, conditions under which that holds are discussed next.



\section{Permutation Invariance of Resource Allocation Policies in WCSs}
\label{subsec:ra_wcss_invariance}
According to Proposition \ref{prop:eisen-regnns-inv}, REGNNs are equivariant to node reorderings; permutation equivariant architectures, however, do not necessarily imply that filter taps learned for a given graph will achieve a similar performance when applied to permutations of the original graph. 
For invariance of the optimal filter taps to hold under permutation, let us first introduce some assumptions on the structure of the  WCS that make the resource allocation problem symmetric.  

\begin{assumption} The one-step cost $c(\cdot)$ in \eqref{eq:constrained_optimal_prob} is permutation invariant, that is for any 
	 $P \in \mathcal{P}$ and any $x \in \mathbb{R}^{m p}$,
\begin{equation}
c(P^\trp x) = c(x).
\label{eq:one-step_invariant}
\end{equation} 
\label{as:one-step_invariant}
\end{assumption}
\begin{assumption}
The constraint function $l(\alpha)$ is permutation invariant in the sense that for any  $P \in \mathcal{P}$ and  $\alpha \in \mathbb{R}^{m}$,
\begin{equation}
l(P^\trp \alpha) = l(\alpha).
\end{equation}
\end{assumption}

\begin{assumption} The control plants have similar dynamics,
\begin{equation}
f^{(1)}(x, u) = \dots  = f^{(m)}(x, u),
\label{eq:similar_dyn}
\end{equation}
and similar control policies,
\begin{equation}
	g^{(1)}(x) = \dots = g^{(m)}(x).
	\label{eq:equivariant_control}
\end{equation}
\label{as:similar_dyn}
\end{assumption}

\begin{assumption} The random variables $w_{o, t}^{(i)}$ in \eqref{eq:observation_noise} are independent and identically distributed (i.i.d.). 
\label{as:iid_obs_noise}
\end{assumption}

Assumption \ref{as:one-step_invariant} is easily satisfied in the quadratic case as long as the weight matrix $Q$ is permutation invariant --- which is equivalent to penalizing the deviation of all plants from the equilibrium point with the same weight.
Requiring all plants to have similar dynamics is a somewhat restrictive but reasonable assumption in industrial applications. 
 Given \eqref{eq:similar_dyn}, condition \eqref{eq:equivariant_control} is  mild and holds for example under the assumption that the controller is a quadratic regulator based on a locally valid linear model of the plants. 
%
%
Under Assumptions \ref{as:one-step_invariant} --- \ref{as:iid_obs_noise}, one can  construct a global function mapping previous plant states to current ones that remains equivariant under node reorderings.
Combined with requirements on the underlying probability distribution characterizing the behavior of the wireless network, this allows us to show that 
 problem \eqref{eq:constrained_optimal_prob} is invariant to permutations. 

\begin{proposition}
Consider the resource allocation problem \eqref{eq:constrained_optimal_prob} and wireless networks defined by probability distributions $\chi(H)$ and $\hat{\chi}(\hat{H})$. 
Assume hypotheses \ref{as:one-step_invariant} --- \ref{as:iid_obs_noise} hold and that 
\begin{equation}
\hat{\chi}(\hat{H}) = \hat{\chi}(P^\trp H P) = \chi(H).
\label{eq:inv_wireless_prob}
\end{equation}
Moreover assume that, for the same permutation matrix, the resource allocation vectors $\alpha$ are such that
\begin{equation}
\alpha(\hat{H}, \hat{x}) = P^\trp \alpha(H,x).
\label{eq:eqv_allocation_vector}
\end{equation}
 Then the optimal cost, $J(\theta^*)$, and the corresponding constraint, $L(\theta^*)$, are invariant under permutations, that is,
\begin{equation}
\begin{aligned}
 &J_{\hat{\chi}(\hat{H})} (\theta^{*}) = J_{\chi(H)} (\theta^{*}), \\
 &L_{\hat{\chi}(\hat{H})} (\theta^{*}) = L_{\chi(H)} (\theta^{*}).
 \label{eq:invariant_opt-cost}
 \end{aligned}
 \end{equation}
\label{prop:cost_invariance}
\end{proposition} 
\begin{proof}
	See appendix \ref{app:proof_prop2}.
\end{proof}
%
The proposition states that, under equivariance assumptions on the probability distributions characterizing the wireless networks, and similar assumptions that make the resource allocation policy and the structure of the WCS symmetric, the optimal cost of the resource allocation problem is invariant under permutations.
That is, reordering the plants --- and the communication graph accordingly --- will not affect the expected optimal cost of the problem. 
The requirement on the equivariance of the resource allocation policy depends on the parameterization in hands and holds for example if the allocation policy is parameterized with REGNNs.
The equivariance of REGNNs  (Proposition \ref{prop:eisen-regnns-inv}) and the invariance of the optimal cost (Proposition \ref{prop:cost_invariance}) support transferability of optimal filters among permutations of the WCS. 
\begin{theorem}
Consider WCSs operating over wireless networks characterized by probability distributions $\chi(H)$ and $\hat{\chi}(\hat{H})$. Assume there exists $P \in \mathcal{P}$ such that, if $\hat{H} = P^\trp H P$ and $\hat{x}_0 = P^\trp x_0$,  the probability distributions satisfy
\begin{equation}
\hat{\chi}(\hat{H}) = \hat{\chi} (P^\trp H P) = \chi (H)
\end{equation}
and Assumptions \ref{as:one-step_invariant} --- \ref{as:iid_obs_noise} hold. Then, a GNN parameter $\theta^*$ is optimal for the wireless control system characterized by distribution $\chi(H)$ if and only if it is optimal for the wireless control system characterized by distribution $\hat{\chi}(\hat{H})$.
\label{theo:filter_invariance}
\end{theorem}

\begin{proof}
Consider a parameter $\theta^*$ that solves \eqref{eq:constrained_optimal_prob} for distribution $\chi(H)$. We have 
\begin{equation}
 \begin{aligned}
  \theta^* \in & \argmin_{\theta}  J(\theta)  \\
  \text{s.t. } &  \mathbb{E} \left[ \sum_{t=0}^T \gamma^t l(\alpha) \right] \leq 0 \\
& \alpha \in \mathcal{A}. 
 \end{aligned}
\end{equation}
If $\hat{\chi}(\hat{H}) = \hat{\chi} (P^\trp H P) = \chi (H)$, Proposition \ref{prop:eisen-regnns-inv} implies that
\begin{equation}
\alpha(\hat{H}, \hat{x}; \theta^*) = P^\trp \alpha(H, x; \theta^*)
\end{equation}
and then Proposition \ref{prop:cost_invariance} holds. Thus, $\theta^*$ induces a feasible solution to the allocation problem defined over $\hat{\chi}(\hat{H})$ and
\begin{equation}
J (\theta^*)_{\hat{\chi}(\hat{H})} = J (\theta^*)_{\chi(H)}.
\end{equation}
Now, consider a GNN parameter $\hat{\theta}^*$ that solves \eqref{eq:constrained_optimal_prob} for distribution $\hat{\chi}(\hat{H})$, and assume that it achieves better performance than $\theta^*$ for $\hat{\chi}(\hat{H})$, that is,
\begin{equation}
    J (\hat{\theta}^*)_{\hat{\chi}(\hat{H})} < J (\theta^*)_{\hat{\chi}(\hat{H})}.
    \label{eq:contradiction_pf_1}
\end{equation}
Under the theorem's assumptions, $\hat{\chi}(\hat{H}) = \hat{\chi} (P^\trp H P) = \chi (H)$. Thus, it follows from Proposition \ref{prop:eisen-regnns-inv} that
\begin{equation}
\alpha(\hat{H}, \hat{x}; \hat{\theta}^*) = P^\trp \alpha(H, x; \hat{\theta}^*),
\end{equation}
and hence Proposition \ref{prop:cost_invariance} holds. That means that $\hat{\theta}^*$ induces a feasible solution for \eqref{eq:constrained_optimal_prob} under distribution $\chi(H)$  with
\begin{equation}
J (\hat{\theta}^*)_{\chi(H)} = J (\hat{\theta}^*)_{\hat{\chi}(\hat{H})}.
\end{equation}
From \eqref{eq:contradiction_pf_1}, however, that would imply
\begin{equation}
    J (\hat{\theta}^*)_{\chi(H)} <  J (\theta^*)_{\chi(H)},
    \label{eq:contradiction_pf_2}
\end{equation}
which violates the optimality of $\theta^*$ for the resource allocation problem defined over $\chi(H)$. Hence, it must hold that, if $\theta^*$ is optimal for the wireless control system characterized by distribution $\chi(H)$, then it is also optimal for the wireless control system characterized by distribution $\hat{\chi}(\hat{H})$. Similar arguments show that the converse also holds, concluding the proof of the theorem.
\end{proof}

\section{Constrained Graph Reinforcement Learning}
\label{sec:res_allocation_RL}


As the actuation signals in the dynamics of the WCS described in equations
 \eqref{eq:ind_plant_dyn} --- \eqref{eq:plant_switched_prob} depend only on the current value of the estimates of the plants states, $x_t$, and we restrict the  allocation policy to be Markovian, the transition of the system to a new state will depend only on the current state of the system, resource allocation decisions and state estimates. 
 We can thus see the resource allocation problem in \eqref{eq:ind_plant_dyn} --- \eqref{eq:constrained_optimal_prob} as a partially observable Markov decision process (MDP), and employ RL algorithms to solve it in a model-free manner.
To account for the constraint on the resource allocation policy, cf. \eqref{eq:constrained_optimal_prob}, we first introduce a dual variable $\lambda \in \mathbb{R}^r$ to formulate the Lagrangian of the constrained RL problem,
\begin{equation}
\begin{aligned}
\mathcal{L} (\theta, \lambda) &\coloneqq J(\theta) + \lambda^\trp \mathbb{E}_{x_0}^{ \alpha(\cdot; \theta)} \left[ \sum_{t=0}^{T} \gamma^t l(\alpha_t) \right] \\
& = \mathbb{E}_{x_0}^{ \alpha(\cdot; \theta)} \left[ \sum_{t=0}^{T} \gamma^t \tilde{r}_t \right],
\end{aligned}
\end{equation} 
with $l(\alpha)$ standing for the resource allocation constraint in \eqref{eq:constrained_optimal_prob} and $\tilde{r}_t$ the penalized one-step cost,
\begin{equation}
\tilde{r}_t = c(x_t) + \lambda^\trp l(\alpha_t).
\end{equation}
The dual optimization problem can then be defined as 
\begin{equation}
D_\theta^* = \max_{\lambda} \min_{\theta \geq 0} \mathcal{L} (\theta, \lambda).
\label{eq:dual_opt_prob}
\end{equation}
With a fixed $\lambda$, the inner optimization problem in \eqref{eq:dual_opt_prob} can  be seen as a standard MDP with objective corresponding to the Lagrangian $\mathcal{L} (\theta, \lambda)$ \cite{lima2020optimal_tsp}. The inner optimization problem can then be solved used standard RL algorithms, while the convexity of the outer optimization problem indicates that the dual variable can be updated via approximate gradient ascent.

As standard policy gradient algorithms can exhibit high variance and slow convergence and value based algorithms are unsuitable to learn continuous allocation decisions, we turn out attention to actor-critic algorithms. In particular, we rely on the Proximal Policy Optimization (PPO) algorithm \cite{schulman2017proximal}, with both the actor and the critic networks parametrized with REGNNs.
The actor network is trained to optimize a surrogate cost function measuring the policy performance while restricting how much the policy can change at each update, 
\begin{equation}
\begin{aligned}
& \theta^* \coloneqq \argmin_{\theta} J_{\text{PPO}}(\theta), \\
&J_{\text{PPO}}(\theta) = \mathbb{E} \left[ \min \left(  r_t(\theta)\hat{A}_t, \text{clip}(r_t (\theta), 1 - \epsilon, 1 + \epsilon ) \hat{A}_t \right) \right],
\end{aligned}
\label{eq:ppo_clip}
\end{equation}
with $\hat{A}_t$ an estimate of the advantage function, defined as the difference between the state-value function $Q_{\alpha(\cdot; \theta)}(s,a)$ --- the expected cost-to-go from state $s$ given action $a$ --- and the expected cost $J_{\alpha(\cdot; \theta)}(s)$ from following the current stochastic policy;
 $\epsilon$ a clipping hyperparameter; and the probability ratio
\begin{equation}
	r_t(\theta) = \frac{\pi_\theta(a_t|s_t)}{\pi_{\theta_{\text{old}}}(a_t|s_t)}.
\end{equation}
The parameters of the critic network are optimized to minimize a squared error loss between the value predicted by the current parameters and the estimated cost-to-go, %
\begin{equation}
\begin{aligned}
\eta^* &:= \argmin_{\eta} \mathbb{E}^{\pi(\cdot; \theta)}_{x_0} (\tdR_t - \rho(s; \eta))^2, \\
\tdR_t & = 
	\sum_{\tau = t}^{ \bar{t} - 1}  \gamma^{\tau - t} \tilde{r}_\tau + \gamma^{t_{\max}} R_{\bar{t}}, t = \bar{t} - 1, \dots, \bar{t} - t_{\max}, 
\label{eq:value_update}
\end{aligned}
\end{equation}
with $\bar{t} = k t_{\max}$, $k \geq 1$, and $t_{\max}$ the interval between updates. $ R_{\bar{t}} = 0$, if $\bar{t} = T$, and $ R_{\bar{t} } = \rho (s_{\bar{t} }; \eta)$ otherwise \cite{schulman2017proximal}. 
The dual variable $\lambda$ can be updated via approximate gradient descent, 
\begin{equation}
\begin{aligned}
\lambda_{i + 1} &= [ \lambda_i + \beta_{\lambda} \hat{\nabla}_{\lambda} \mathcal{L}(\theta, \lambda) ]_{+}, \\
\hat{\nabla}_{\lambda} \mathcal{L}(\theta, \lambda) &= \mathbb{E}^{\pi(\cdot; \theta)} \left[ \sum_{t = 0}^T l(\alpha_t) \right].
\end{aligned}
\label{eq:dual_update}
\end{equation}

 As described in Algorithm \ref{alg:rl}, we rely on $N$ simultaneous realizations of the WCS with fixed horizon $T$. Initial states of the control plants are sampled from a standard normal distribution, and fading states representing the transmission conditions of the wireless channel are sampled from $\chi(H)$. At each time step, the allocation decisions for each realization of the WCS are sampled from the agent's policy $\pi(\cdot; \theta)$, the plants states evolve according to equations \eqref{eq:ind_plant_dyn} --- \eqref{eq:plant_switched_prob}, and updated fading states are sampled from $\chi(H)$. Every $t_{\max}$ time steps, $N t_{\max}$ transitions are then used to update the policy and value networks according to \eqref{eq:ppo_clip} and \eqref{eq:value_update}, respectively. At the end of each training episode, $NT$ transitions are used to estimate the dual gradient, $\hat{\nabla}_\lambda \mathcal{L}(\theta, \lambda)$, and the dual variable $\lambda$ is updated according to \eqref{eq:dual_update}.

\begin{algorithm}[t]
	\caption{Graph RL for WCSs}
	\label{alg:rl}
	\DontPrintSemicolon
	\KwResult{Resource allocation policy.}
	\BlankLine
	\For{$i$ = $1, \dots, E_{\text{RL}}$}{
		\For{ t = 1, \dots, T}{
			Collect transitions from $N$ realizations \;
			\If{$t = k t_{\max}, k = 1, 2, \dots$}{ 
				Update policy network \eqref{eq:ppo_clip}\;
				Update value network \eqref{eq:value_update}\;
			}
			Update dual variable \eqref{eq:dual_update} \;
		}
	}
\end{algorithm}

\begin{remark}
Training can be facilitated by initially teaching the allocation policy to imitate a heuristic policy using some imitation learning method, such as the Dataset Aggregation (DAgger) algorithm \cite{pmlr-v15-ross11a}. Once the agent has successfully learned to imitate the heuristic policy, it can then use Algorithm \ref{alg:rl} to improve the resource allocation policy.
\end{remark}

\section{Numerical experiments}
\label{sec:num_exp}



\begin{figure*}[htb]
    \centering
    \begin{minipage}{0.33\textwidth}
        \centering
        \includegraphics[width=\linewidth]{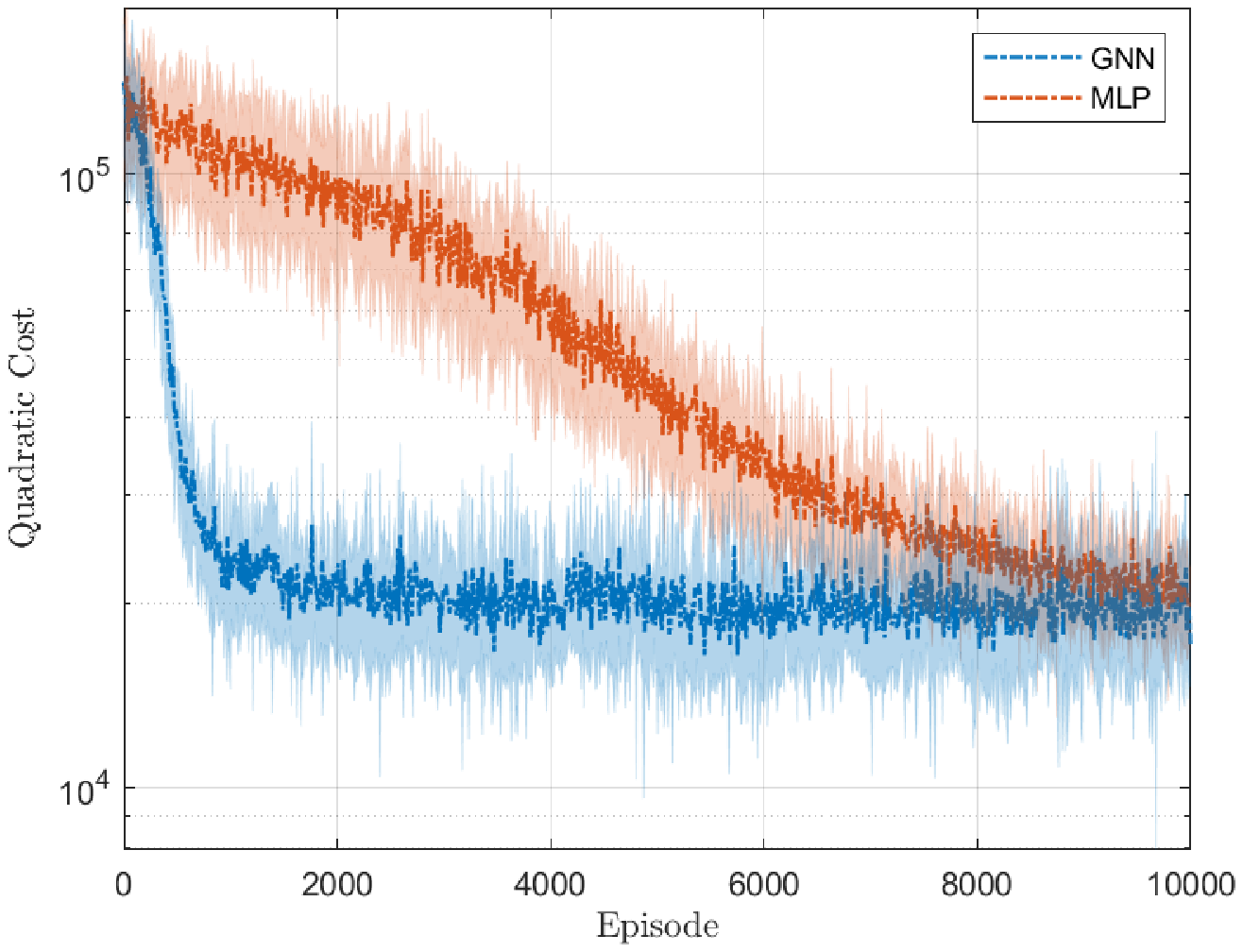}
        \caption{Multi-cellular network: training cost.}
\label{fig:scheduling_training_ppo}
    \end{minipage} 
        \hfill %
        \begin{minipage}{0.33\textwidth}
        \centering
       \includegraphics[width=\linewidth]{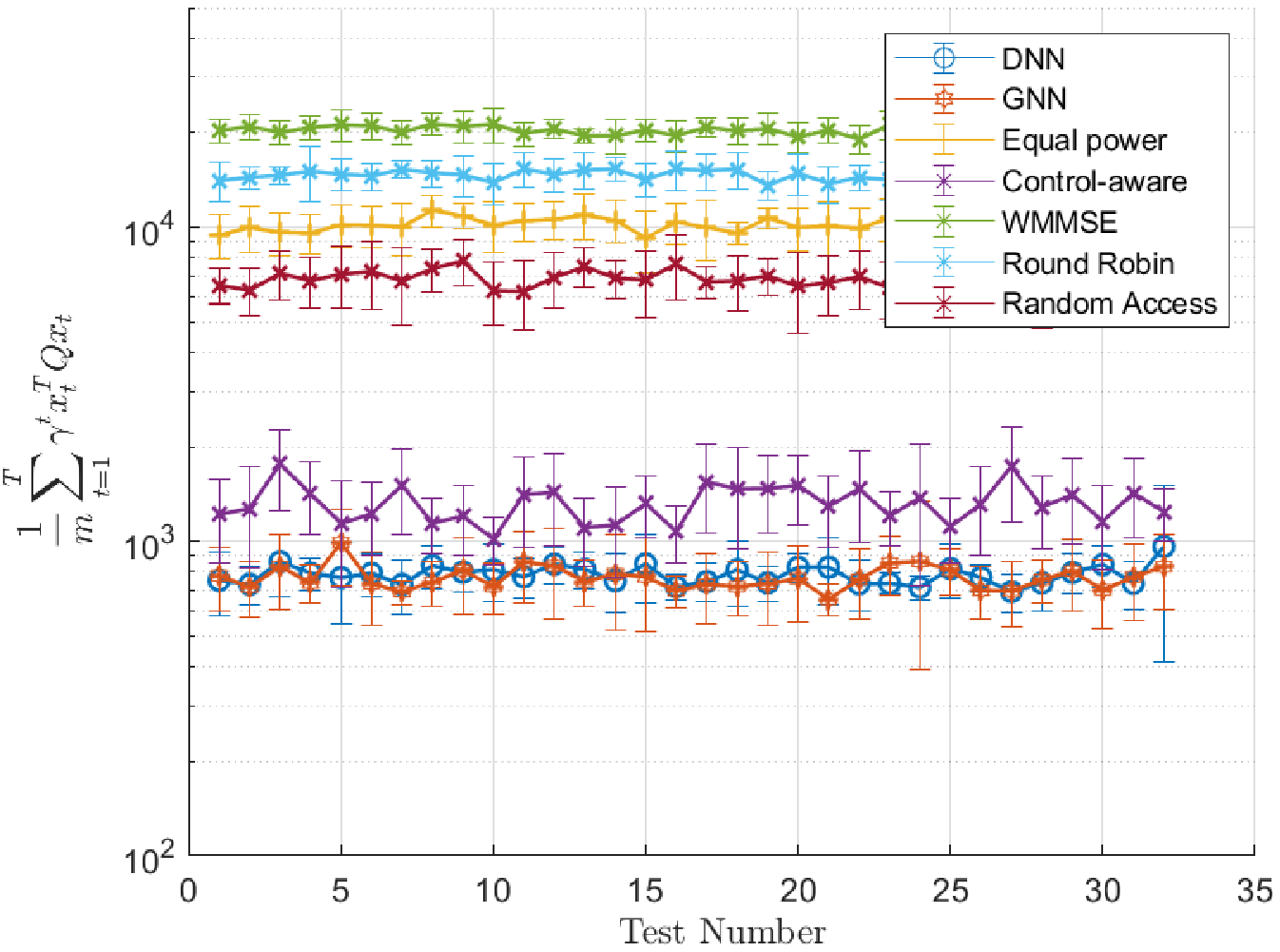}
       \caption{Multi-cellular network: runtime cost.}
\label{fig:scheduling_test_ppo}
    \end{minipage}%
\hfill %
\begin{minipage}{0.33\textwidth}
	\centering
	\includegraphics[width=\linewidth]{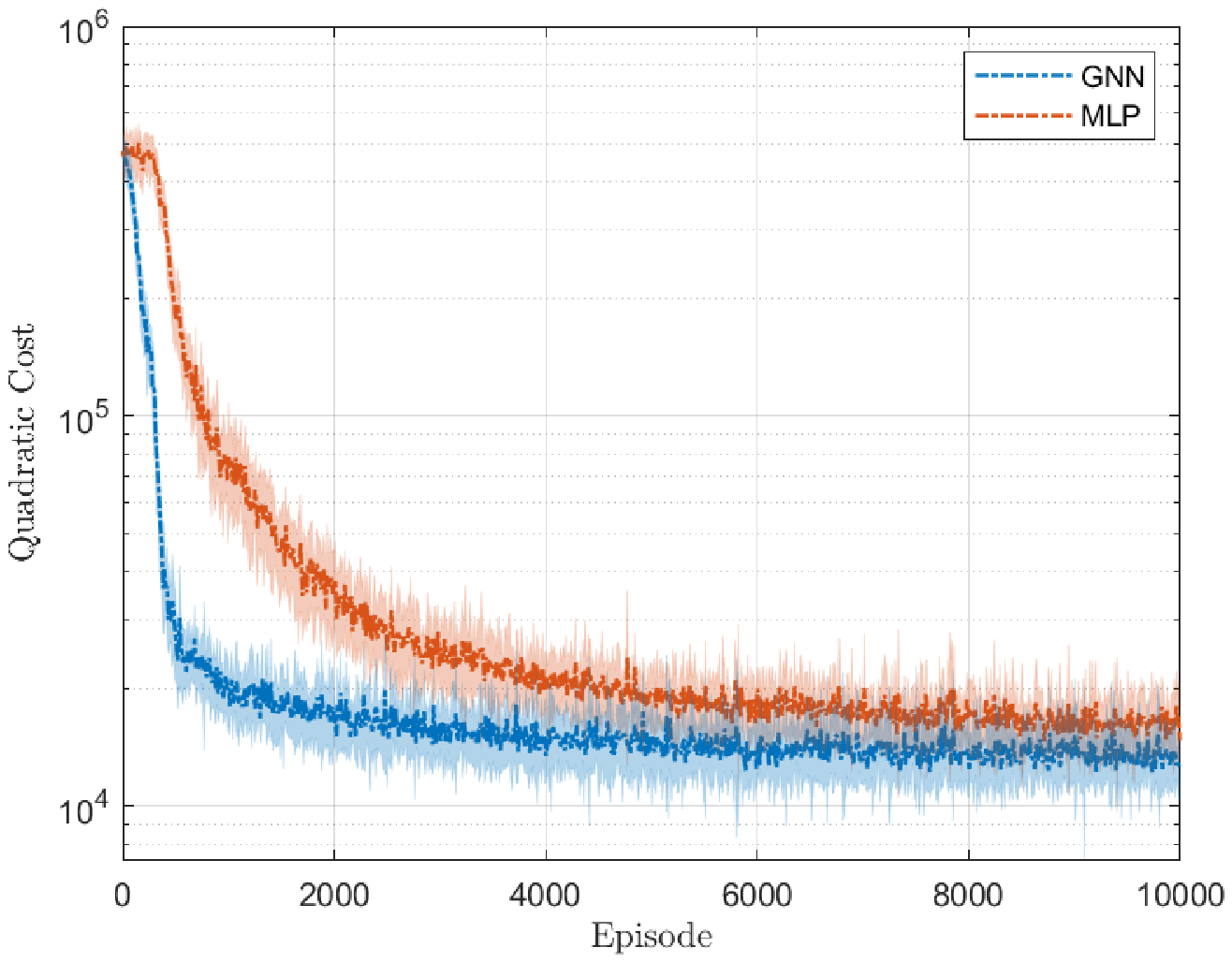}
	\caption{Distributed decisions: training cost.}
	\label{fig:base_scheduling_training_ppo}
\end{minipage}%
\end{figure*}


Next we present numerical experiments to compare the performance of the graph RL approach presented in Section \ref{sec:res_allocation_RL} against standard deep RL techniques and heuristics commonly used to design resource allocation policies for WCSs. We created custom environments on OpenAI Gym \cite{brockman2016openai} and customized the Stable Baselines 3 PPO implementation \cite{stable-baselines3} to account for the dual learning step in Algorithm \ref{alg:rl}. Unless indicated otherwise, graph RL policies were parametrized with a REGNN consisting of 3 graph convolutional layers with 5 filter taps and 10 features per layer, and deep RL policies were parametrized with a standard multilayer neural network made up by two hidden layers with 64 hidden units each.  To facilitate training, the resource allocation policies receive the norm of the observations of each plant as inputs.
 Policies were trained over $N = 16$ simultaneous realizations of the WCS, with a discount factor $\gamma = 0.95$, imitation learning rate $\beta_{\text{IL}} = 5 \times 10^{-4}$,  reinforcement learning step size $\beta_{\text{RL}} = 5 \times 10^{-5}$ and dual variable step size $\beta_\lambda = 1 \times 10^{-5}$.  Hyperparameters were  tuned heuristically.

We consider that the plants in the WCS are all linear and open-loop unstable, and have the same dynamics,
\begin{equation}
x_{t + 1}^{(i)} = Ax^{(i)}_t + B u_t^{(i)} + w_t,
\label{eq:linear_plants}
\end{equation}
with
\begin{equation}
A = \begin{bmatrix}
1.05 & 0.2 & 0.2 \\ 0 & 1.05 & 0.2 \\ 0 & 0 & 1.05 
\end{bmatrix}; \quad
B = \mathbb{I}.
\label{eq:matrices_value_power_allocation}
\end{equation}
Plants are spatially  distributed around the remote base stations. The path loss coefficient, $p_l$, and the distance $d_i$ between a plant and its controller determine the slow fading component, $h_s = d_i^{-p_l}$. The fast fading component $h_f$ is randomly sampled from a Rayleigh distribution $\chi(h)$ with parameter $\lambda_h$, leading to
$
h_{ij} = h_s^{(ij)} \times h_f^{(ij)}.
$ Given resource allocation decisions $\alpha$ and fading and interference conditions $H$, the probability of each plant successfully closing its feedback loop --- unknown to the agent --- is given by 
\begin{equation}
v(\xi^{(i)}(\alpha, H)) = 1 - \exp (- \xi^{(i)}(\alpha, H)),
\end{equation}
with $\xi^{(i)}$ the SINR for plant $i$.

\subsection{Multi-Cellular Networks}\label{sec_numericals_cellular}

In this first experiment we consider a multi-cellular network with $n = 5$ base stations and $k = 6$ users per base station. The base stations are evenly spaced in a line, with the distance between consecutive base stations given by $m/n$. The plants are randomly placed around the base stations, with the vertical position of a plant sampled uniformly from an interval $[-k, k]$ around the corresponding base station, and its horizontal position sampled uniformly from an interval $[-m/2n, m/2n]$ around the corresponding base station.  
In this scenario, the resource allocation policy must decide whether to send or not the control signal to a particular plant. The total power budget available at each time step, $mp_0$, is divided equally between the transmitting signals. Here, we take $Q = \mathbb{I}$, $\lambda_h = 2$, $p_l = 1.5$, $p_0 = 2.5$ and $T = 50$. We do not use imitation learning to pre-train the agents in this scenario. The deep RL and the graph RL agents are trained over 10000 episodes each.
Figure \ref{fig:scheduling_training_ppo} shows the mean and standard deviation of the finite horizon cost per training episode for the graph RL (blue) and the deep RL (orange) policies over the $N = 16$ parallel realizations of the WCS. Although both policies achieve a similar performance, the REGNN policy  converges after about 2000 episodes, while the standard deep RL resource allocation policy requires almost 5 times as many episodes to converge. 

\begin{figure*}[htb]
	\centering
	\begin{minipage}{0.33\textwidth}
		\centering
		\includegraphics[width=\linewidth]{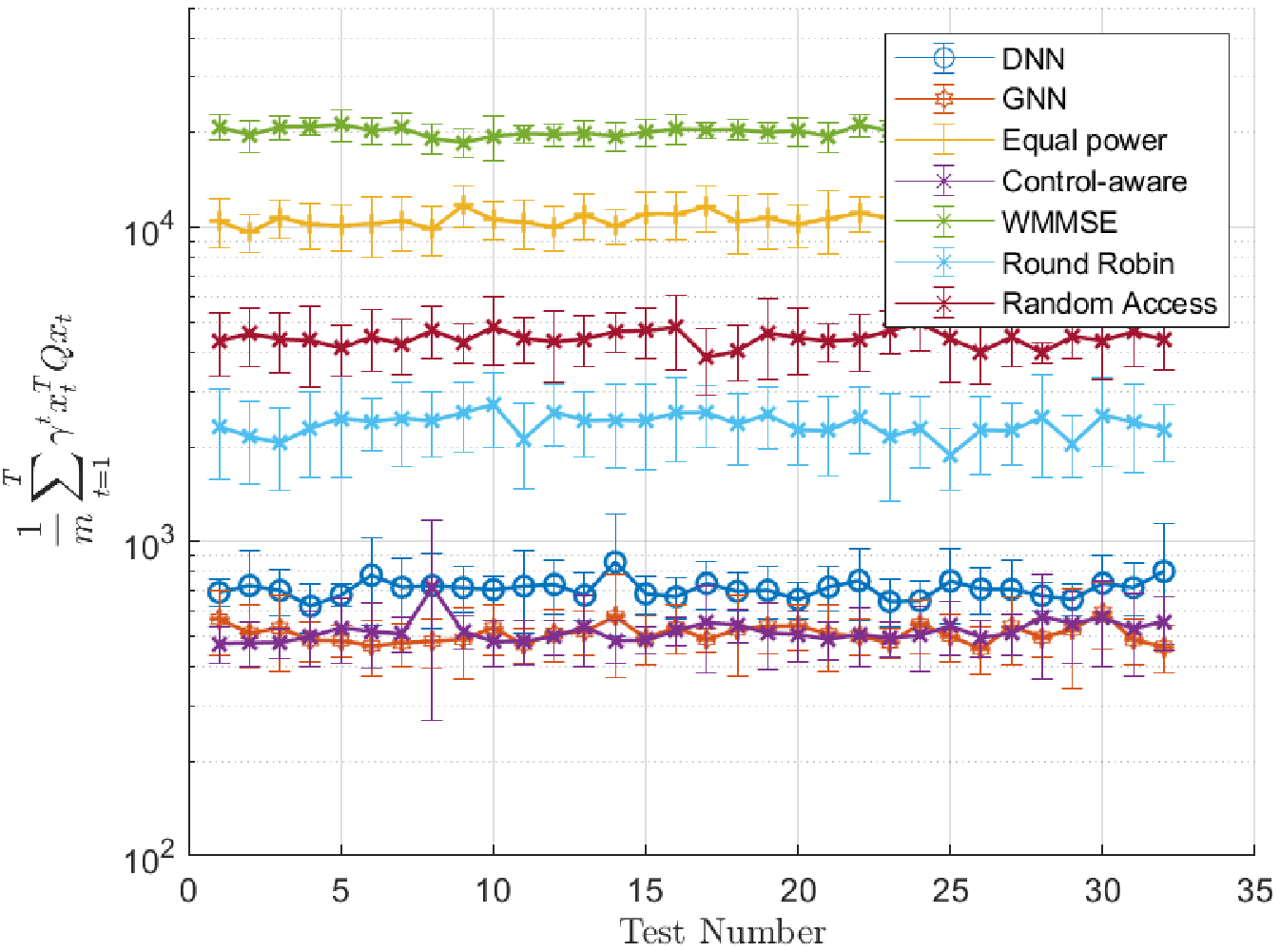}
		\caption{Distributed decisions: runtime cost.}
		\label{fig:base_scheduling_test_ppo}
	\end{minipage}%
	\hfill %
	\begin{minipage}{0.33\textwidth}
		\centering
		\includegraphics[width=\linewidth]{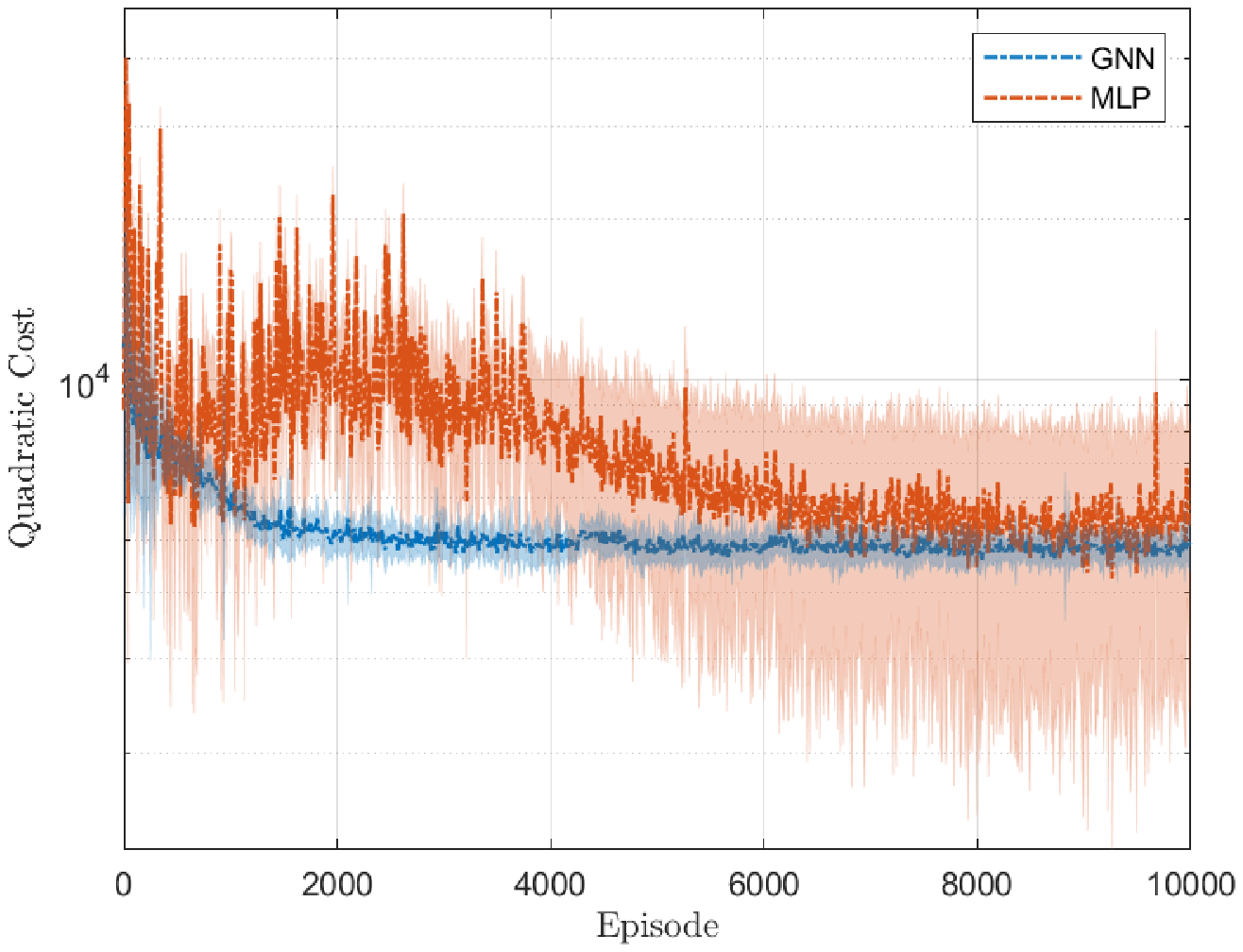}
		\caption{Ad-hoc network: training cost.}
		\label{fig:adhoc_constraint_cost}
	\end{minipage} %
	\hfill %
	\begin{minipage}{0.33\textwidth}
		\centering
		\includegraphics[width=\linewidth]{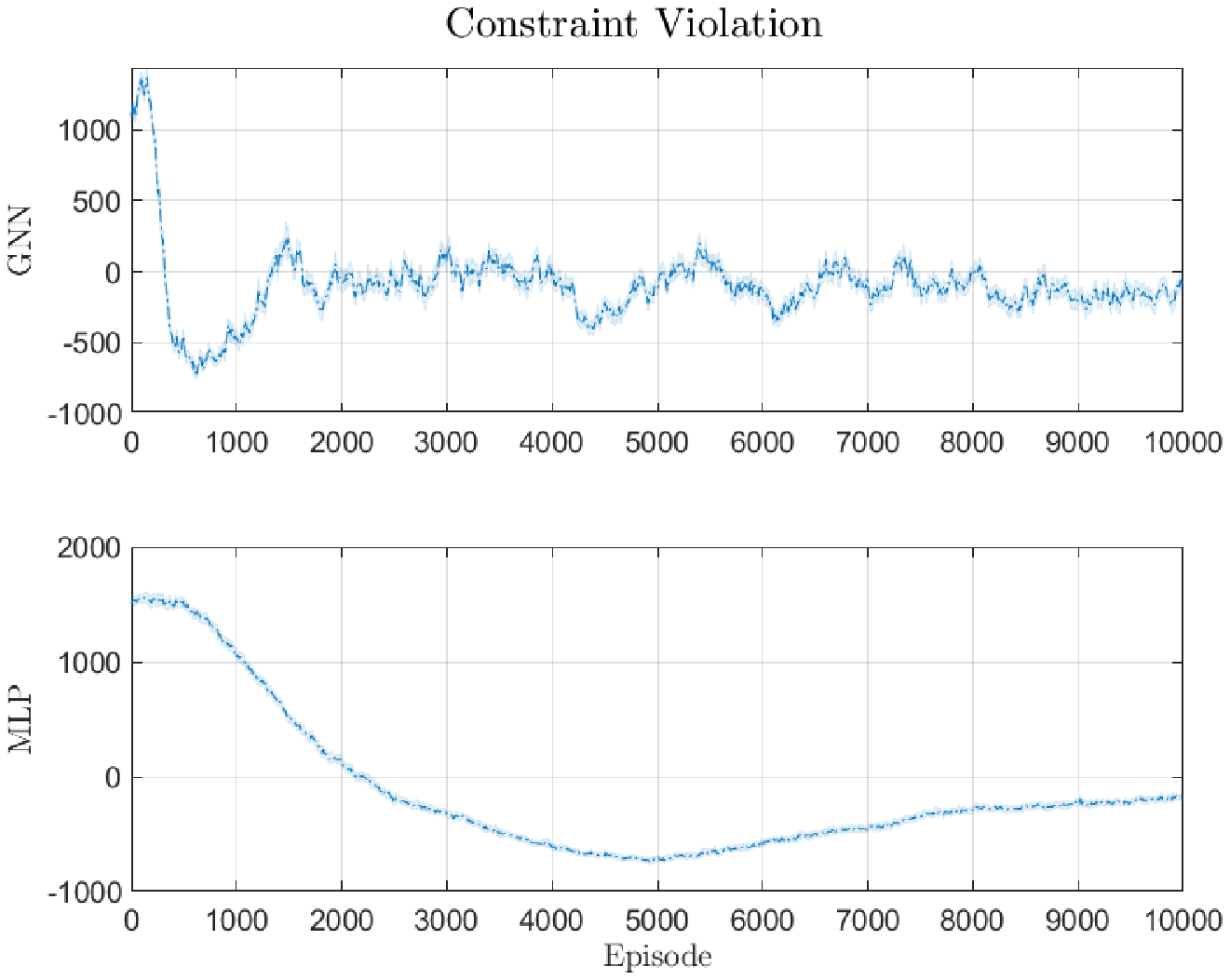}
		\caption{Ad-hoc network: constraint violation.}
		\label{fig:adhoc_constraint_constraintviolation}
	\end{minipage}%
\end{figure*}

After training, we compare the performance of the graph RL and deep RL allocation policies against some commonly used heuristics, namely dividing power equally among the plants; WMMSE,  the weighted minimum mean squared error algorithm \cite{shi_wmmse_2011} --- commonly used to maximize the weighted sum rate (WSR) in wireless channels subject to power constraints --- and standard scheduling policies:
\begin{enumerate}
\item Control-aware: chooses $m/3$ plants currently further away from equilibrium point to transmit;
\item Round robin: schedules  $m/3$ plants per  time instant;
\item Random access: randomly chooses  $m/3$ plants;
\end{enumerate}

To make comparisons between the learned allocation policies and the solutions outlined above fair, the total transmitting power, $mp_0$, is divided equally between the transmitting signals. We consider a longer horizon, $T = 80$, than the one seen by the agents during training. The finite horizon cost achieved by each solution, scaled by the number of plants in the network, $m$, is shown in Figure \ref{fig:scheduling_test_ppo}. Each test point corresponds to the mean over 10 realizations, with each realization using different random seeds. The graph RL and the deep RL policies achieve the same performance in this scenario, and both outperform the heuristic solutions mentioned above. The resource allocation policy parametrized by REGNNs, however, can be trained faster, and is transferable at scale --- as we show in Section \ref{subsec:transference_sims}.

We next consider a scenario similar to the first experiment, but with the allocation decisions now distributed between the base stations. More precisely, each base station now chooses one of the plants in its vicinity to transmit to.  Figure \ref{fig:base_scheduling_training_ppo} shows the mean and standard deviation of the finite horizon cost per training episode for the graph RL (blue) and the deep RL (orange) policies in this scenario. Here, not only does the resource allocation policy parametrized by REGNNs converge faster, but it also achieves a smaller operation cost than the one achieved by the deep RL solution. 
We adapt the heuristic solutions to this scenario, and once again the total transmitting power, $mp_0$, is divided equally between the transmitting signals. We consider a longer horizon, $T = 80$, than the one seen by the agents during training. The finite horizon cost achieved by each solution (scaled by the number of plants in the network) is shown in Figure \ref{fig:base_scheduling_test_ppo}.
In this scenario, the graph RL policy and the control-aware heuristic outperform the other solutions.

\subsection{Ad-Hoc Networks}\label{sec_numericals_adhoc}
\label{subsec:adhoc_sims}


We now consider an ad-hoc network with $m = 30$ pairs of plants and remote controllers. Controllers are evenly spaced in a line, with the distance between consecutive controllers equals to 4. Plants are randomly placed in an $[-m/10, m/10]^2$ area around the corresponding controller. Agents aim to stabilize the control plants while respecting a sum-power constraint, 
\begin{equation}
 \mathbb{E} \left[ \sum_{t=0}^T \gamma^t \left( \sum_{i=1}^m  \alpha^{(i)}_{t} - mp_0\right)  \right] \leq 0.
\end{equation}
In this scenario, $p_0 = 2.5$, $\lambda_h=2$, and $p_l = 1.5$.
Each training episode has a fixed horizon $T = 30$. First, the resource allocation policies are trained to imitate the WMMSE algorithm using the DAgger algorithm \cite{pmlr-v15-ross11a} for a total of 1000 episodes. After the imitation learning phase, we then use reinforcement learning updates to train the resource allocation policies parametrized by REGNNs and by standard neural networks, with the training curves shown in Figures \ref{fig:adhoc_constraint_cost} and \ref{fig:adhoc_constraint_constraintviolation}. The mean and standard deviation of the quadratic objective per training episode is shown in Figure \ref{fig:adhoc_constraint_cost}, where one can see that the graph RL approach converges faster than the deep RL policy. The graph RL resource allocation policy also quickly converges to a feasible solution, as shown in Figure \ref{fig:adhoc_constraint_constraintviolation}.

After training, we next compare the performance of the graph RL and deep RL allocation policies against heuristics adapted to this setting, with each of the scheduling heuristics choosing $m/3$ plants to transmit at each time instant.
To make comparisons fair,  the total transmitting power, $mp_0$, is divided equally between the transmitting signals for the heuristic solutions, and we use a softmax layer to ensure that the total power used by the learned solutions at each time step is equal to $mp_0$. We consider a longer horizon, $T = 80$, than the one seen by the agents during training. The finite horizon cost achieved by each solution, scaled by the number of plants in the network, $m$, is shown in Figure \ref{fig:adhoc_constraint_test_cost}. The REGNN policy outperforms both the deep RL and the heuristic solutions. 

\begin{figure*}[htb]
    \centering
    \begin{minipage}{0.32\textwidth}
		\centering
		\includegraphics[width=\linewidth]{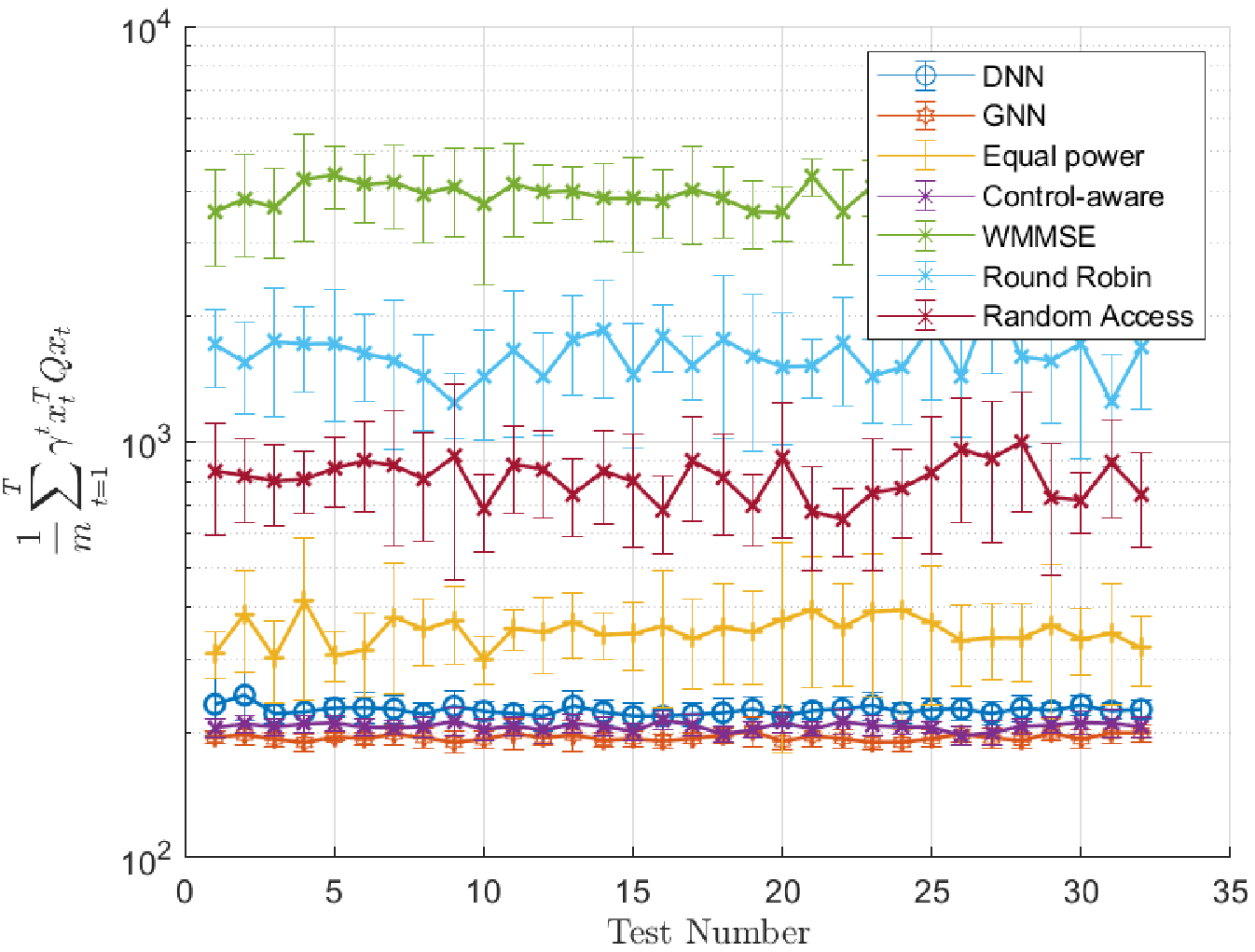}
		\caption{Ad-hoc network: runtime cost.}
		\label{fig:adhoc_constraint_test_cost}
	\end{minipage} %
    \hfill %
    \begin{minipage}{0.32\textwidth}
        \centering
    	\includegraphics[width=\linewidth]{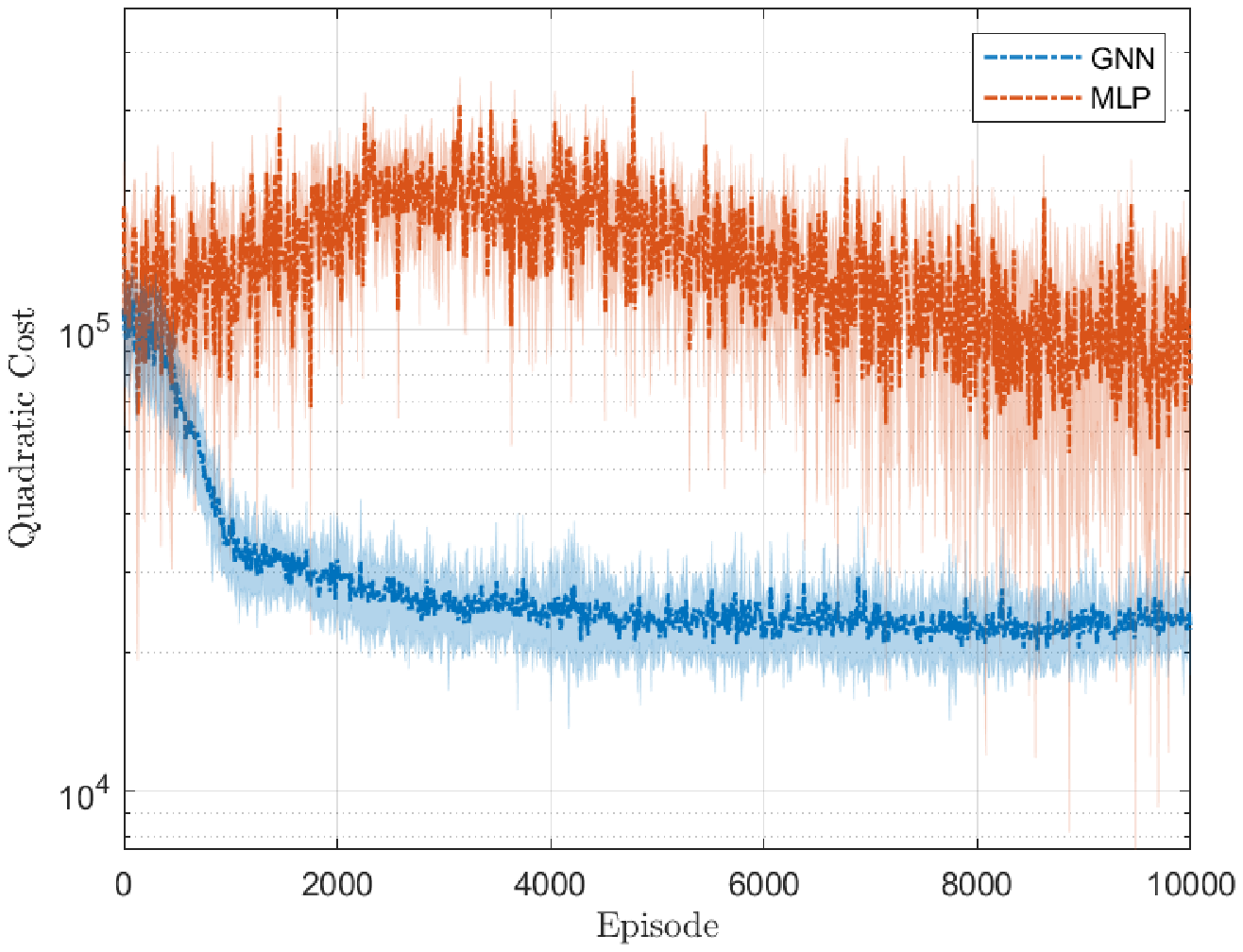}
    	\caption{Larger network: training cost.}
		\label{fig:adhoc_constraint_cost_60p}
    \end{minipage} %
	\hfill %
	\begin{minipage}{0.32\textwidth}
    	\centering
    	\includegraphics[width=\linewidth]{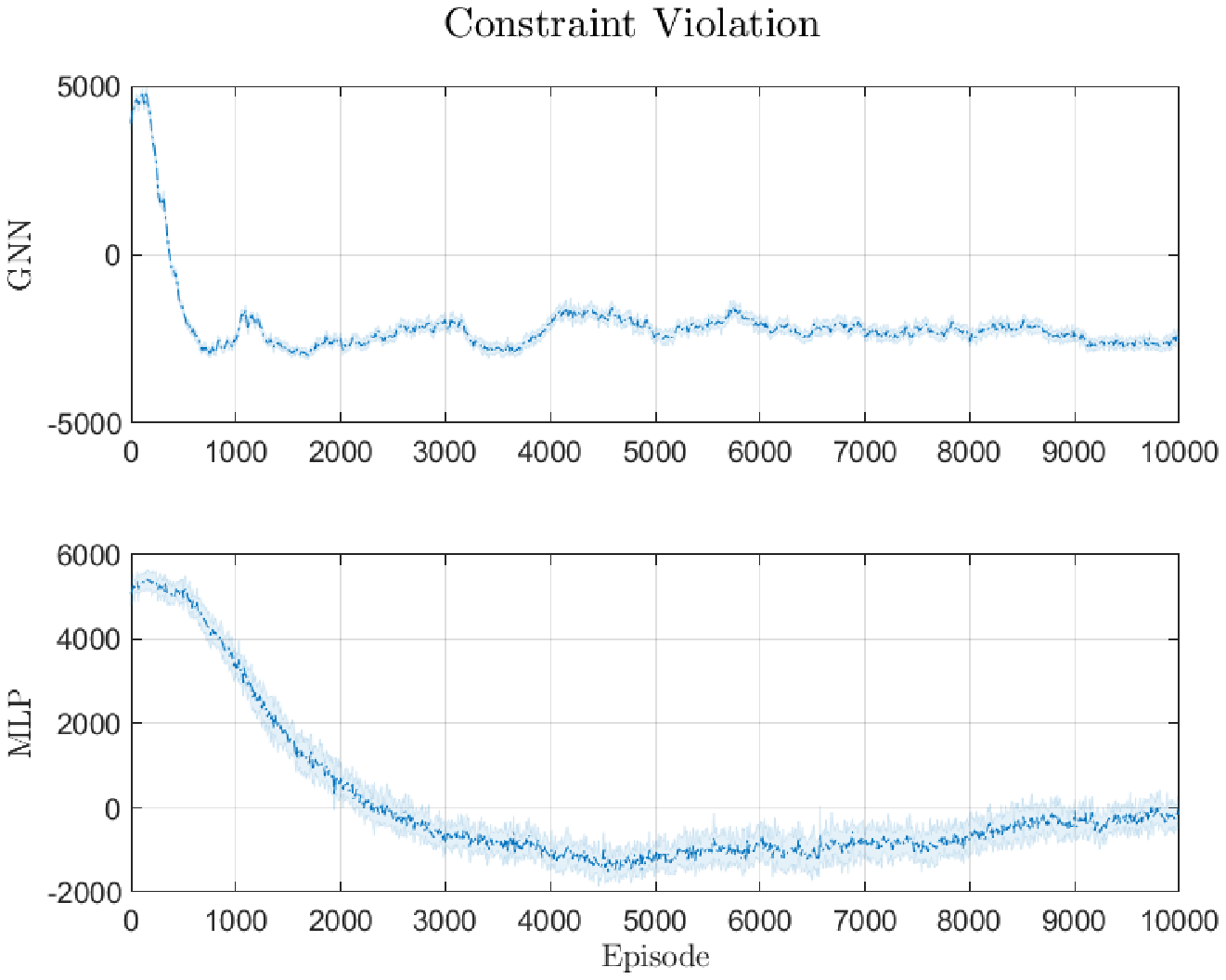}
    	\caption{Larger network: constraint violation.}
		\label{fig:adhoc_constraint_constraintviolation_60p}
    \end{minipage}%
\end{figure*}
For a larger network with $m = 60$ pairs of plants and remote controllers, the performance of the deep RL approach degrades, as shown in Figures \ref{fig:adhoc_constraint_cost_60p} --- \ref{fig:adhoc_constraint_statetraj_60p}. As in the previous scenario, the remote controllers are evenly spaced in a line, with the distance between consecutive remote controllers equals to 4. Plants are randomly placed in an $[-m/10, m/10]^2$ area around the corresponding controller and we have $p_0 = 5$, $\lambda_h=2$, and $p_l = 1.5$. As shown in Figures \ref{fig:adhoc_constraint_cost_60p} and \ref{fig:adhoc_constraint_constraintviolation_60p}, both approaches still find feasible allocation policies, but the resource allocation policy parametrized with REGNNs achieves a better training objective than the resource allocation policy parametrized with standard neural networks. The runtime simulations (Figures \ref{fig:adhoc_constraint_testcost_60p} and \ref{fig:adhoc_constraint_statetraj_60p}) also consider a longer horizon ($T = 80$) than the one seen during training ($T = 30$). The graph RL policy achieves a smaller cost of operation than the deep RL and heuristic solutions, as shown in Figure \ref{fig:adhoc_constraint_testcost_60p}. Each test point  in Figure \ref{fig:adhoc_constraint_testcost_60p} shows the mean and standard deviation of the overall quadratic cost (scaled by the number of plants) over 10 realizations under different random seeds.
As expected from the smaller cost of operation, the GNN allocation policy maintains the plants closer to the equilibrium point, as the trajectory of control states under the graph RL, deep RL, WMMSE and control-aware resource allocation policies for one of the test realizations, in Figure \ref{fig:adhoc_constraint_statetraj_60p}, shows. 

\begin{figure*}[tb]
    \centering
    \begin{minipage}{0.32\textwidth}
		\centering
		\includegraphics[width=\linewidth]{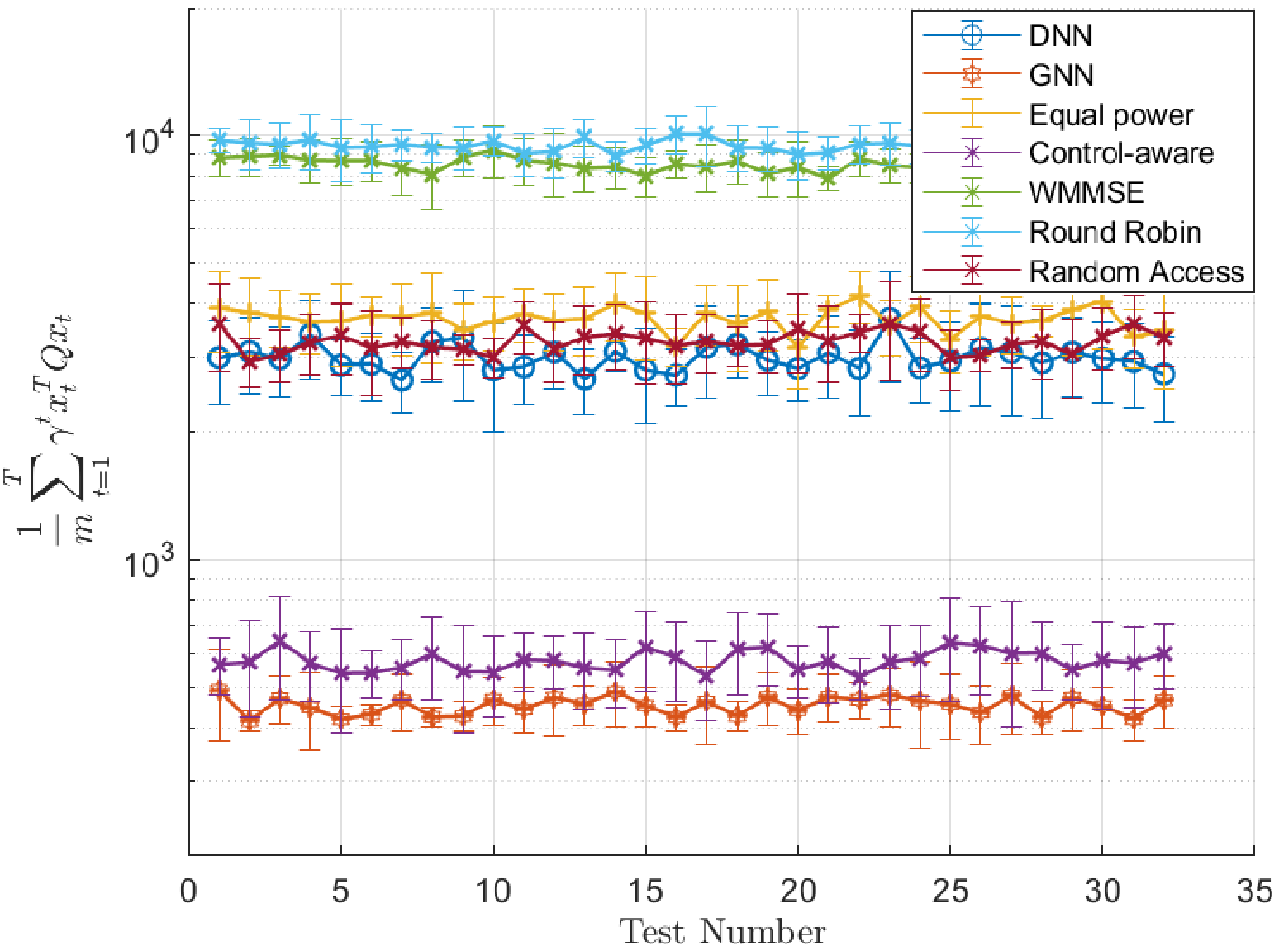}
		\caption{Larger network: runtime cost.}
		\label{fig:adhoc_constraint_testcost_60p}
	\end{minipage}%
    \hfill %
    \begin{minipage}{0.32\textwidth} 
		\centering
		\includegraphics[width=\linewidth]{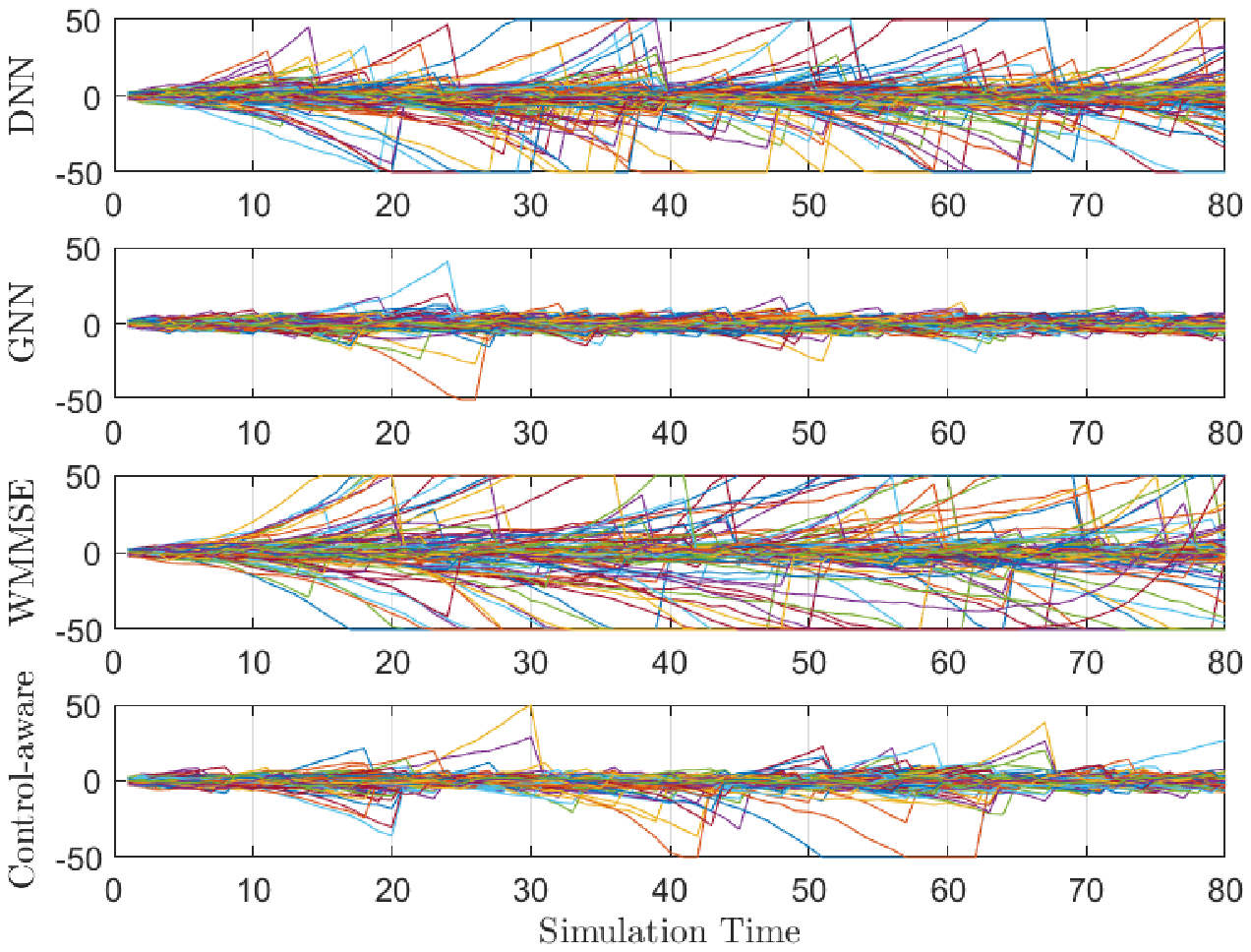}
		\caption{Trajectory of control states.}
		\label{fig:adhoc_constraint_statetraj_60p}
	\end{minipage}  %
	\hfill %
	\begin{minipage}{0.32\textwidth}
		\centering
		\includegraphics[width=\linewidth]{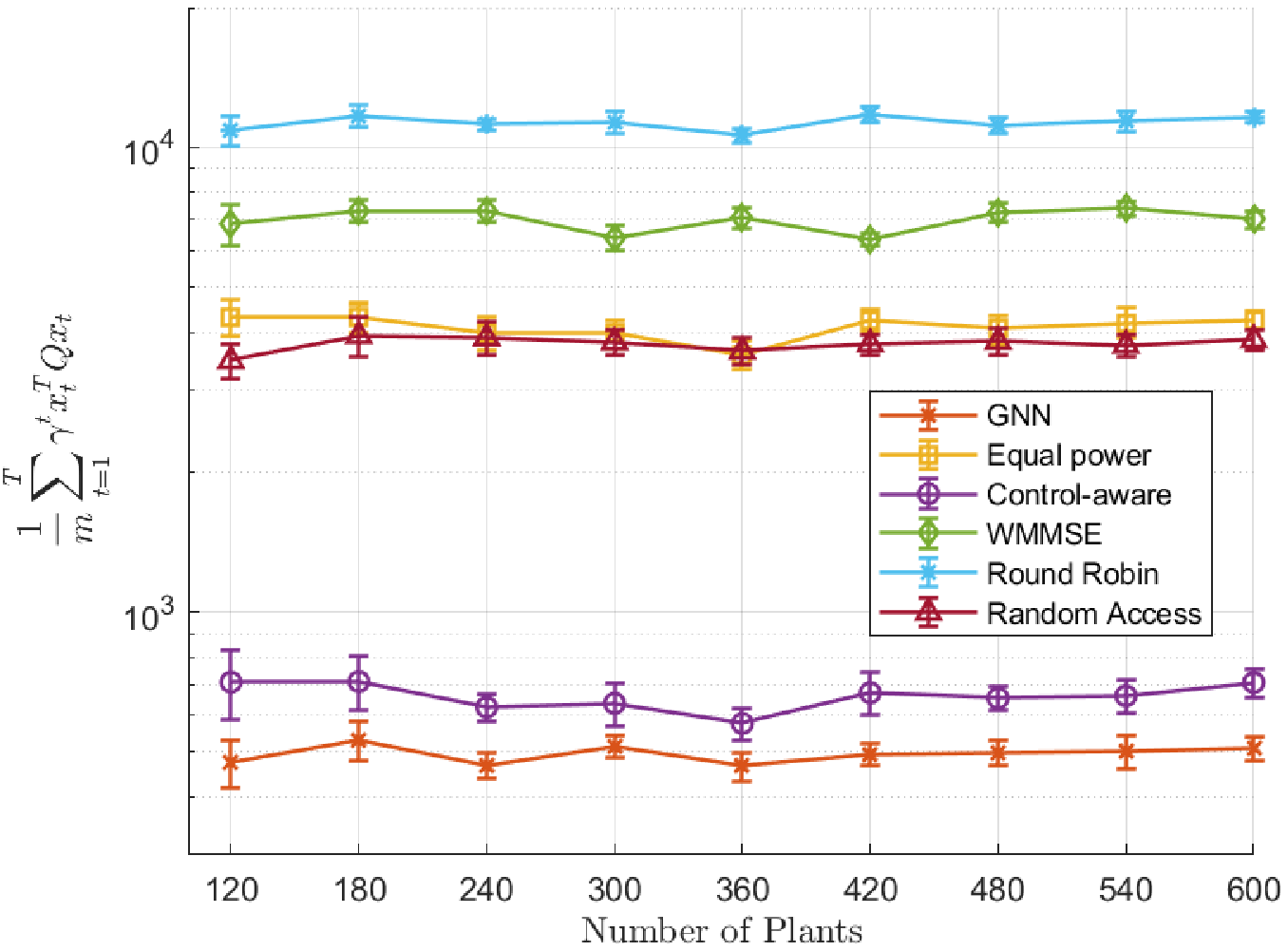}
		\caption{Ad-hoc network: transference test.}
		\label{fig:adhoc_constraint_test_transf_cost}
	\end{minipage} %
\end{figure*}

\subsection{Transference}
\label{subsec:transference_sims}


The trainable parameters of a REGNN correspond to the filter taps used to modulate the signal received from successive $k$-hop neighborhoods of each node, as seen in equations \eqref{eq:graph_conv} - \eqref{eq:gnn_multilayer}. Given a time-varying GSO $S_t$, the filter taps of a REGNN do not depend on the number of users in the network, and can thus be transferred to networks of varying size --- as long as the number of features of the input signal remains the same. That implies that one can train a resource allocation policy on a certain network, and transfer the learned parameters to be executed on a larger network. To evaluate how well the performance of a REGNN resource allocation policy transfers \emph{at scale} to larger networks, we revisit the ad-hoc scenario with $m = 60$ pairs of plants and remote controllers, and draw successively larger ad-hoc networks with controllers evenly spaced in a line and the distance between consecutive remote controllers equals to 4. The plants are randomly place in an $[-6, 6]^2$ area around the corresponding controller, and we take  $p_0 = 5$, $\lambda_h=2$, and $p_l = 1.5$, as in the previous scenario. We then compare the performance of the REGNN policy against the heuristic solutions over a runtime horizon $T = 80$, and present the simulation results in Figure \ref{fig:adhoc_constraint_test_transf_cost}. Each test point shows the mean and standard deviation of 20 realizations of the WCS over a network of a certain size. As shown in Figure \ref{fig:adhoc_constraint_test_transf_cost}, the REGNN allocation policy consistently outperforms the heuristic solutions on networks up to 10 times larger the network over which it was trained.
Note that a standard neural network is not transferable at scale, hence we do not evaluate its performance in this scenario.

\section{Conclusion}


This paper presents a constrained graph reinforcement learning approach to design resource allocation policies for large-scale wireless control systems. Resource allocation problems are challenging, even in the absence of policies tailored to the specific application in mind, such as wireless control systems. To tackle this problem and learn or improve upon allocation policies, previous works have successfully relied on deep learning and deep reinforcement learning. As neural networks are made up of successive linear computational layers followed by pointwise nonlinearities, however, those approaches may fail to scale. In this paper, we then propose the use of reinforcement learning and \emph{graph} neural networks to design feasible, scalable resource allocation policies for wireless control systems. Extensive numerical experiments demonstrate that the proposed approach yields resource allocation policies that routinely outperform deep RL and heuristic solutions, and are transferable across networks of varying size. 


\appendices
\section{Proof of Proposition \ref{prop:cost_invariance} }
\label{app:proof_prop2}


The invariance of the optimal cost is a direct consequence of the equivariance of the resource allocation policy and of the invariance of the cost and constraint functions. For the permuted system defined over the wireless network described by the probability distribution $\hat{\chi}(\hat{H})$, the cost is given by
\begin{equation}
\begin{aligned}
J_{\hat{\chi}(\hat{H})}(\theta) &= \mathbb{E}^{\hat{\alpha}(\hat{x}, \hat{H}; \theta)} \left[ \sum_{t = 0}^T \gamma^t c(\hat{x}_t) \right] \\
&= \int_{(\hat{\mathcal{S}} \times \hat{\mathcal{A})}^T} \left( \sum_{t = 0}^T \gamma^t c(\hat{x}_t)   \right) \hat{p}_{\hat{\alpha}}(\mathbf{\hat{s}}, \mathbf{\hat{\alpha}}) d \mathbf{\hat{s}} d \mathbf{\hat{\alpha}},
\end{aligned}
\label{eq:theorem_proof_cost1}
\end{equation}
with $\mathbf{\hat{s}} = (\hat{s}_0, \hat{s}_1, \dots)$ and $\mathbf{\hat{\alpha}} = (\hat{\alpha}_0, \hat{\alpha}_1, \dots)$. Since the resource allocation policy is assumed to be Markovian, it follows that 
\begin{multline}
\hat{p}_{\hat{\alpha}}(\mathbf{\hat{s}}, \mathbf{\hat{\alpha}})
= \prod_{u = 1}^T \hat{p} (\hat{s}_u |\hat{s}_{u - 1}, \hat{\alpha}_{u-1} ) \pi(\hat{\alpha}_u |\tilde{ \hat{s}}_u) \hat{p}_o (\tilde{ \hat{s}}_u|  \hat{s}_u) \\ \hat{p}(\hat{s}_0) \pi(\hat{\alpha}_0|\hat{s}_0) \hat{p}_o (\tilde{ \hat{s}}_0|  \hat{s}_0),
\label{eq:joint_distribution_1}
\end{multline}
with $\hat{p}(\hat{s}_0)$ standing for the distribution of the initial state,  $\hat{p} (\hat{s}_u |\hat{s}_{u - 1}, \hat{\alpha}_{u-1} )$ representing a one-step transition of the system, $\hat{p}_o (\tilde{ \hat{s}}_u|  \hat{s}_u)$ standing for the distribution of the observation noise and $\tilde{ \hat{s}}_u$ the observation received by the agent.

Now introduce the change of variables $\hat{x} = P^\trp x, \hat{H} = P^\trp H P$. To simplify the notation, let $s_t = [x_t, H_t]$ and, with a slight abuse of notation, $P^\trp s_t = [P^\trp x_t, P^\trp H_t P]$ in the following. 
Since the wireless network characterizing the permuted system satisfies \eqref{eq:inv_wireless_prob} and the permutation matrices are orthogonal, we have that
\begin{equation}
\hat{p}(\hat{s}_0) = \hat{p}(P^\trp s_0) = p(s_0).
\end{equation}
Similarly, it follows from the assumption on the permutation equivariance of the resource allocation policy that
\begin{equation}
\pi(\hat{\alpha}_t | \tilde{\hat{s}}_t) = \pi(P^\trp \alpha_t | P^\trp \tilde{s}_t) = \pi(\alpha_t|\tilde{s}_t).
\end{equation}
Under Assumption \ref{as:iid_obs_noise},
\begin{equation}
\hat{p}_o (\tilde{ \hat{s}}_u|  \hat{s}_u) = p_o  (\tilde{ s}_u| s_u).
\end{equation}
The kernel $\hat{p}(\hat{s}_t|\hat{s}_{t-1}, \hat{\alpha}_{t-1}) \, : \, \hat{\mathcal{S}} \times (\hat{\mathcal{S}},\hat{\mathcal{A}}) \to [0,1]$ representing a one-step transition of the permuted WCS is given by
\begin{equation}
\begin{aligned}
\hat{p}(\hat{s}_t|\hat{s}_{t-1}, \hat{\alpha}_{t-1}) &= \hat{p} (\hat{x}_t, \hat{H}_t| \hat{x}_{t-1}, \hat{H}_{t-1}, \hat{\alpha}_{t-1}) \\
&= \hat{p} (\hat{x}_t| \hat{x}_{t-1}, \hat{H}_{t-1}, \hat{\alpha}_{t-1}) \hat{\chi}(\hat{H}_t) \\
\end{aligned}
\label{eq:transition_kernel_00}
\end{equation}
since $\hat{H}_t$ is independent of $\hat{x}_{t-1}, \hat{H}_{t-1}$, and $ \hat{\alpha}_{t-1}$, and thus,
\begin{equation}
\begin{aligned}
&\hat{p}(\hat{s}_t|P^\trp s_{t-1}, P^\trp \alpha_{t-1}) \\
&= \hat{p} (\hat{x}_t| P^\trp x_{t-1}, P^\trp H_{t-1} P, P^\trp \alpha_{t-1}) \hat{\chi}(P^\trp H_t P) \\
&= \hat{p} (\hat{x}_t| P^\trp x_{t-1}, P^\trp H_{t-1} P, P^\trp \alpha_{t-1}) \chi(H_t)
\end{aligned}
\label{eq:transition_kernel_01}
\end{equation}
by \eqref{eq:inv_wireless_prob}. Now, let 
\[\hat{f}(\hat{x},\hat{u}) \coloneqq [\hat{f}^{(1)}(\hat{x}^{(1)}, \hat{u}^{(1)}), \dots, \hat{f}^{(m)}(\hat{x}^{(m)}, \hat{u}^{(m)})]
\] 
aggregate the dynamics of individual plants given the control signals received under the communication model \eqref{eq:plant_switched_prob},
 $\mathbb{I}_{X}$ represent the indicator function for a region $X \subset \mathcal{X}$, and $\hat{\xi} = \xi(P^\trp \alpha, P^\trp H P) = P^\trp \xi(\alpha, H)$. Let also 
$\hat{p}_{u} (\hat{u}|\hat{\xi})$ represent the distribution of the received control signal $\hat{u}$ given the effective SINR $\hat{\xi}$, and $\hat{p}_w(\hat{w})$ represent the distribution of the control plants noise for the permuted system --- and accordingly, $p_{u} (u|\xi), p_w(w)$ for the unpermuted one. Then,
\begin{equation}
\begin{aligned}
&\hat{p}(\hat{x}_{t+1} \in P^\trp X|P^\trp s_{t}, P^\trp \alpha_{t}) \\
&= \mathbb{E} \left[ \mathbb{I}_{P^\trp X} \left( \hat{f}(P^\trp x_{t}, \hat{u}_t )  + \hat{w}  \right) \right] \\
&=  \int  \mathbb{I}_{P^\trp X} \left( \hat{f}(P^\trp x_{t}, \hat{u}_t)  + \hat{w}   \right) \hat{p}_{u} (\hat{u}|\hat{\xi})  \hat{p}_w(\hat{w}) d \hat{u} d\hat{w} \\
&=  \int  \mathbb{I}_{P^\trp X} \left( \hat{f}(P^\trp x_{t}, P^\trp u_t)  + P^\trp w  \right) \\ &\qquad \qquad \hat{p}_u (P^\trp u|P^\trp \xi) \hat{p}_{w}(P^\trp w) d u dw \\
&=  \int  \mathbb{I}_{X} \left(f(x_{t}, u_t)  + w  \right) p_u (u|\xi) p_w(w) d u dw  \\
&= p(x_{t+1} \in X|s_{t}, \alpha_{t}) 
\end{aligned}
\label{eq:Q_hat_1}
\end{equation}
 since the collection of Bernoulli random variables representing the packet drops over the communication channels for the WCS are independent, the control policies are all the same according to Assumption \ref{as:similar_dyn}, and the permutation matrices are orthogonal. 
Combining \eqref{eq:joint_distribution_1} --- \eqref{eq:Q_hat_1}, 
\begin{equation}
\hat{p}_{\hat{\alpha}}( \mathbf{\hat{s}}, \mathbf{\hat{\alpha}}) 
= p_{\alpha} (\mathbf{s}, \mathbf{\alpha})
\end{equation}
and thus
\begin{multline}
J_{\hat{\chi}(\hat{H})} (\theta)  
= \int_{(\hat{\mathcal{S}} \times \hat{\mathcal{A})}^T} \left( \sum_{t = 0}^T \gamma^t c(\hat{x}_t)   \right) \hat{p}_{\hat{\alpha}}(\mathbf{\hat{s}}, \mathbf{\hat{\alpha}}) d \mathbf{\hat{s}} d \mathbf{\hat{\alpha}}, \\
= \int_{(\mathcal{S} \times \mathcal{A})^T} \left( \sum_{t = 0}^T \gamma^t c(P^\trp x_t)   \right) p_{\alpha}(\mathbf{s}, \mathbf{\alpha}) d \mathbf{s} d \mathbf{\alpha}  \\
= \int_{(\mathcal{S} \times \mathcal{A})^T} \left( \sum_{t = 0}^T \gamma^t c( x_t)    \right) p_{\alpha}(\mathbf{s}, \mathbf{\alpha}) d \mathbf{s} d \mathbf{\alpha} \\
= J_{\chi(H)} (\theta)
\label{eq:theorem_proof_cost2}
\end{multline}
by the invariance of the one-step cost in Assumption \ref{as:one-step_invariant}.
 Now let $L_{\chi(H)} (\theta)$ the constraint achieved by the resource allocation policy. Following similar arguments, we have
\begin{multline}
L_{\hat{\chi}(\hat{H})}(\theta) = \mathbb{E}^{\hat{\alpha}(\hat{x}, \hat{H}; \theta)} \left[ \sum_{t = 0}^T \gamma^t l(\hat{\alpha}_t(\hat{x}_t, \hat{H}_t; \theta)) \right]  \\
= \int_{(\hat{\mathcal{S}} \times \hat{\mathcal{A})}^T} \left( \sum_{t = 0}^T \gamma^t l(\hat{\alpha}_t)   \right) \hat{p}_{\hat{\alpha}}(\mathbf{\hat{s}}, \mathbf{\hat{\alpha}}) d \mathbf{\hat{s}} d \mathbf{\hat{\alpha}} \\
= \int_{(\mathcal{S} \times \mathcal{A})^T} \left( \sum_{t = 0}^T \gamma^t l(P^\trp \alpha_t)   \right) p_{\alpha}(\mathbf{s}, \mathbf{\alpha}) d \mathbf{s} d \mathbf{\alpha} \\
= L_{\chi(H)} (\theta).
\label{eq:theorem_proof_constraint_1}
\end{multline}
 
 Finally, note that the relations above hold for any parameterization that renders the resource allocation policy equivariant under permutations. In particular, it also holds for the parameter set 
 that solves \eqref{eq:constrained_optimal_prob}, from which \eqref{eq:invariant_opt-cost} follows.

\bibliographystyle{IEEEtran}
\bibliography{wl_control}

\end{document}